\documentclass[journal]{IEEEtran}
\usepackage{amsmath}
\usepackage{amssymb}
\usepackage{amsfonts}
\usepackage{graphicx}
\usepackage[top=0.52in, bottom=0.52in, left=0.52in, right=0.52in]{geometry}
\usepackage{enumerate,float,color,graphicx,fancybox,pifont,epsf,epsfig,subfigure,amsmath,amssymb,psfrag,amsthm}
\usepackage{pstool} 
\usepackage{algorithm}
\usepackage{algorithmic}
\usepackage{cite}
\newcounter{MYtempeqncnt}

\newtheorem{ex}{\textbf{Example}}
\newtheorem{theorem}{\textbf{Theorem}}
\newtheorem{lemma}[theorem]{\textbf{Lemma}}
\newtheorem{proposition}[theorem]{\textbf{Proposition}}
\newtheorem{remark}[theorem]{\textbf{Remark}}

\newcommand{\secref}[1]{Section~\ref{#1}}

\newcommand{\figref}[1]{Figure~\ref{#1}}

\newcommand{\theoref}[1]{Theorem~\ref{#1}}

\newcommand{\proref}[1]{Proposition~\ref{#1}}
\newcommand{\lemref}[1]{Lemma~\ref{#1}}
\newcommand{\remref}[1]{Remark~\ref{#1}}

\newcommand{\procref}[1]{Procedure~\ref{#1}}

\newcommand{\subparagraph}{}
\usepackage[compact]{titlesec}
\usepackage[belowskip=-6pt,aboveskip=4pt]{caption}
\usepackage[font=scriptsize]{caption}

\def\Tr{\mathrm{Tr}}

\title{Power-Constrained Sparse Gaussian Linear Dimensionality Reduction over Noisy Channels}
\author{Amirpasha Shirazinia, \textit{Member, IEEE} and Subhrakanti Dey, \textit{Senior Member, IEEE}}
\begin{document}
\maketitle
\thispagestyle{empty}
\begin{abstract}
In this paper, we investigate power-constrained sensing matrix design in a sparse Gaussian linear dimensionality reduction framework. Our study is carried out in a single--terminal setup as well as in a multi--terminal setup consisting of orthogonal or coherent multiple access channels (MAC). We adopt the mean square error (MSE) performance criterion for sparse source reconstruction in a system where source-to-sensor channel(s) and sensor-to-decoder communication channel(s) are noisy. Our proposed sensing matrix design procedure relies upon minimizing a lower-bound on the MSE in single-- and multiple--terminal setups. We propose a three-stage sensing matrix optimization scheme that combines semi-definite relaxation (SDR) programming, a low-rank approximation problem and power-rescaling. Under certain conditions, we derive closed-form solutions to the proposed optimization procedure. Through numerical experiments, by applying practical sparse reconstruction algorithms, we show the superiority of the proposed scheme by comparing it with other relevant methods. This performance improvement is achieved at the price of higher computational complexity. Hence, in order to address the complexity burden, we present an equivalent stochastic optimization method to the problem of interest that can be solved approximately, while still providing a  superior performance over the popular methods. \let\thefootnote\relax\footnote{A. Shirazinia (email: amirpasha.shirazinia@signal.uu.se) and S. Dey (email: subhrakanti.dey@signal.uu.se) are with Signals and Systems Division, Department of Engineering Sciences, Uppsala University, Uppsala Sweden.} 
\let\thefootnote\relax\footnote{This paper was partially presented at IEEE International Conference in Communications (ICC) 2015.}
\end{abstract}


\section{Introduction} \label{sec:intro}
Sensor networks have recently attracted much research interest due to their practical popularity in accomplishing autonomous tasks, such as monitoring, sensing, computation and communication. Diverse applications of sensor networks motivate the deployment of new techniques and algorithms due to systems' limited resources, computational complexity and power consumption. In this regard, compressed sensing (CS) \cite{06:Donoho,06:Candes,08:Candes} can be considered as an emerging tool for signal compression and acquisition that significantly reduces costs due to sampling, leading to low-power consumption and low-bandwidth communication.

CS is a framework for simultaneous signal acquisition and compression, which is based on linear dimensionality reduction. The CS framework guarantees accurate (or, even exact) signal recovery from far fewer number of acquired measurements, under the condition that the source signal can be represented by a sparse form. Indeed, CS builds upon the fact that many types of physically-observed signals (such as voice, image, etc.) can be represented by only \text{a few} few non-zero components in a known basis, where these few components convey the most informative portion the signal. 

In order to clarify the concept of CS in relation to the objectives of our work, let us consider the linear reduction model $\mathbf{y = A x + n}$, where $\mathbf{x} \in \mathbb{R}^N$ is a sparse signal (in a known basis\footnote[1]{In a more precise manner, the CS measurement vector is written as $\mathbf{y = A x + n}$, where $\mathbf{x}$ is a non-sparse input vector. We assume that $\mathbf{x}$ has a sparse representation $\boldsymbol{\theta}$ in a known basis $\mathbf{\Psi}$ such that $\mathbf{x} = \mathbf{\Psi} \boldsymbol{\theta}$. Then, the CS measurement equation can be written as $\mathbf{y} = \mathbf{A \Psi} \boldsymbol{\theta} + \mathbf{n}$. Hence, if $\mathbf{\Psi}$ is known at the time of reconstruction, the original non-sparse vector $\mathbf{x}$ can be recovered from the reconstruction of the sparse vector $\boldsymbol{\theta}$ directly. In this paper, for simplicity of presentation, and without loss of generality, we assume that $\mathbf{\Psi}$ is equivalent to the identity transform, and therefore $\mathbf{x}$ is sparse.}) vector with a size higher than that of the measurement vector $\mathbf{y} \in \mathbb{R}^{M}$. Further, $\mathbf{A} \in \mathbb{R}^{M \times N}$ is a \textit{fat} sensing matrix (i.e., $M < N$), and $\mathbf{n} \in \mathbb{R}^{M}$ is the measurement noise vector. 

For the purpose of reconstructing the sparse vector from the CS measurements, several techniques have been developed based on convex
optimization methods (see e.g.~\cite{06:Candes2,07:Candes}), iterative greedy search algorithms (see
e.g. \cite{07:Tropp,08:Blumensath,09:Dai,12:Saikat}) and Bayesian estimation approaches (see
e.g. \cite{07:Larsson,08:Ji,09:Elad,10:Protter,12:Kun}). It should be mentioned that a careful design of the sensing matrix $\mathbf{A}$ is crucial in order to achieve good performance of sparse reconstruction algorithms. Moreover, as shown in \cite{14:Shirazi,13:Pasha-icassp2}, the sensing matrix has an important role not only  in determining the amount of estimation error, but  also in deciding  the amount of distortion due to quantization and transmission of CS measurements over digital communication channels. Therefore, in this paper, we are interested in the optimized design of the sensing matrix $\mathbf{A}$ with respect to an appropriate performance criterion. Regarding the theory and applications of CS, sensing matrices are generally divided into two main groups: deterministic or random. Although most early work in CS was based on stochastic sensing matrix generation, such matrices are often not feasible in practice for hardware implementations \cite{11:Duarte}. Motivated by this fact, we focus on deterministic sensing matrices, and show that an optimized design of a sensing matrix can substantially improve the performance of CS. 


\subsection{Background} \label{sec:background}
In the literature, available approaches for designing deterministic sensing matrices for estimation purposes can be divided into (but not limited to) three broad categories as described below.

	1) In the first category, the sensing matrix design is linked to a fundamental feature of the sensing matrix $\mathbf{A}$, called mutual coherence \cite{01:Donoho}, which is defined as follows
\begin{equation} \label{eq:mutual co}
    \mu \triangleq  \underset{i \neq j}{\max} \hspace{0.2cm} \frac{|\mathbf{A}_i^\top \mathbf{A}_j|}{\|\mathbf{A}_i\|_2 \|\mathbf{A}_j\|_2}, \hspace{0.2cm} 1 \leq i,j \leq N,
\end{equation}
where $\mathbf{A}_i$ denotes the $i^{th}$ column of $\mathbf{A}$. For a sensing matrix, a smaller value of the mutual coherence is desired in order for the matrix to behave similar to an orthogonal transform. The notion of mutual coherence is important since many worst-case performance guarantee bounds developed for sparse reconstruction algorithms often build upon its quantity (see e.g., \cite{10:Haim}). One of the early works within this category is \cite{07:Elad} that studied the optimal design of sensing matrix in the sense of reducing the mutual coherence (or average mutual coherence for average signal recovery performance).
	
	2) In the second category, in order to analytically address the sensing matrix design problem in a more  tractable manner, the sensing matrix $\mathbf{A}$ is optimized by minimizing the Frobenous--norm distance  between the Gram matrix of the sensing matrix (or, the product of the sensing matrix and a given matrix) and an identity matrix. This method, indeed, reveals how far the sensing matrix can be from an orthogonal transform. Formally, in this line of work, the following optimization problem is posed under relevant constraints: 
	\begin{equation} \label{eq:2nd cat}
	\begin{aligned}
		&\underset{\mathbf{A}}{{\text{minimize}}} \hspace{0.25cm} \| \mathbf{\Psi}^\top \mathbf{A}^\top \mathbf{A} \mathbf{\Psi}  - \mathbf{I}_N \|_F ,& \\
	\end{aligned}
	\end{equation}
 where $\| \cdot \|_F$ denotes the Frobenius norm and $\mathbf{\Psi}$ is a known matrix (e.g., a sparsifying dictionary) with appropriate dimension. Although the optimal sensing matrix with respect to minimizing \eqref{eq:2nd cat} does not necessarily minimize the mutual coherence, it has been shown that, using this method, the mutual coherence of the sensing matrix can be considerably reduced. Some examples within this category are \cite{11:Zelnik,10:vahid,10:vahid2,13:Gang}. Further, in \cite{09:Duarte}, simultaneous optimization of sensing matrix and sparsifying dictionary has been studied which follows the ideas in \cite{11:Aharon}. 
	
	3) While in the first and second categories, the sensing matrix is designed to address the worst-case performance of sparse reconstruction, the actual performance, such as estimation error or mean square error (MSE) of sparse source reconstruction, can be typically far less. Exploiting randomness in the sparse source vector, one might consider minimizing
\begin{equation} \label{eq:MSE cat3}
	\mathrm{MSE} \triangleq \mathbb{E}[\|\mathbf{x} - \widehat{\mathbf{x}}\|_2^2],
\end{equation}
under relevant constraints. Here, $\| \cdot \|_2$ denotes the $\ell_2$--norm, and $\widehat{\mathbf{x}}$ represents the output of decoder (e.g., a linear or non-linear estimator, a sparse reconstruction algorithm, etc.) at the receiving end. MSE is one of the most commonly-used criteria of accuracy for estimation and reconstruction purposes. Adopting the MSE as a targeted performance criterion in CS systems has called for redeveloping classical Bayesian methods for sparse reconstruction which have been extensively studied recently in \cite{08:Shihao,09:Elad,10:Protter,11:Turek,12:Peleg,14:Turek,11:Zhilin}. Optimizing sensing matrix with respect to minimizing the MSE is not only effective in improving the performance of Bayesian-based sparse reconstruction algorithms, but also of other types of sparse reconstruction algorithms, such as greedy search or convex algorithms. In \cite{12:Chen}, the authors proposed a two-stage optimization procedure in order to design a sensing matrix with respect to minimizing a lower-bound on the reconstruction MSE of a sparse source with known statistical properties. In the context of linear dimensionality reduction models with linear decoding, the authors in \cite{08:Jin,07:Schizas} have investigated optimized design of sensing matrices in a decentralized (multi--terminal) setting, where reconstruction MSE of a given (not necessarily sparse) source with known covariance matrix is considered subject to an average transmit power constraint. Also, Yuan \textit{et. al.} in \cite{14:Yuan} has studied the same optimization problem, in a single--terminal setup, under linear decoding, but by constraining the volume of error covariance matrix instead of a total power constraint.
	
\subsection{Contributions} \label{sec:contributions}
Our contributions, in this paper, lie in the third category described above. In particular, they are as described below: 

	\textit{i. Single--terminal Scenario:} We consider a correlated Gaussian sparse source vector (i.e., the non-zero components of the source signal are correlated Gaussian random variables), that is scaled linearly and subsequently corrupted by additive noise before compression/encoding via a CS-based sensing matrix. The resulting CS measurements are transmitted over a noisy (analog) communication channel, modeled by channel gain and additive noise, under an available average transmit power constraint. At the receiving-end, the source signal is decoded using an estimator (e.g., linear or non-linear estimator, sparse reconstruction algorithm, etc.)  to reconstruct the sparse source. 
	
	\textit{ii. Multi--terminal Scenario:} We consider a correlated Gaussian sparse source vector that is scaled linearly and  corrupted by additive noise, via separate terminals prior to compression/encoding via CS-based sensing matrices. The CS measurement vectors are transmitted over orthogonal or coherent multiple access channels (MAC), under an available average transmit power constraint. The fusion center (FC), at the receiving-end, decodes the sparse source signal.

In the above scenarios, we aim at optimizing the sensing matrix (or, matrices) by minimizing a \textit{lower-bound} on the MSE incurred by using the MMSE estimator (which by definition minimizes the MSE) of a sparse source signal. We adopt the MSE of the oracle MMSE estimator as the lower-bound on the MSE to be minimized under an average transmit power constraint. We propose a three-stage sensing matrix optimization procedure that combines semi-definite relaxation (SDR) programming, a low-rank approximation problem and power-rescaling. The solution to the low-rank approximation problem can be derived analytically, and the SDR programming problem can be solved using convex optimization techniques. Further, in the multi--terminal settings with orthogonal and coherent MAC, we formulate and solve convex optimization problems in order to optimally rescale the power. Under certain conditions, we derive closed-form solutions to the proposed optimization procedure. For example, in the single-terminal scenario, we analytically show that if the non-zero components of the sparse source are uncorrelated, and the source-to-sensor channels are perfect, then the optimal solutions to the three-stage optimization procedure are tight frames\footnote[1]{Formally, a frame is defined as a sequence of column vectors $\mathbf{A}_i$ of a matrix $\mathbf{A}$, and the frame is said to be tight if the associated  matrix $\mathbf{A}$ has a singular-value decomposition (SVD) of the form $\mathbf{U}_a [\mathbf{I}_M \;  \; \mathbf{0}_{M \times (N-M)}] \mathbf{V}_a^\top$, where $\mathbf{U}_a$ and $\mathbf{V}_a$ are unitary matrices with appropriate dimensions.} \cite{08:Kovac}, which are easy to construct, and play important roles in signal processing, denoising, coding, etc. Through numerical experiments, by applying practical sparse reconstruction algorithms, we compare our proposed scheme with other relevant methods. Experimental results show that the proposed approach improves the MSE performance by a large margin compared to other methods. This performance improvement is achieved at the price of higher computational complexity which arises from the fact that the objective function, i.e., the lower-bound, sweeps over all possible sparsity patterns of the source. In order to tackle the complexity issue, we develop an equivalent stochastic optimization method to the problem of interest, which can be approximately solved, while still providing a superior performance over the competing methods.

Our sensing matrix design for the oracle estimator is different with that of \cite{12:Chen} in the sense that we minimize the oracle MMSE estimator under a power constraint, while in \cite{12:Chen} the oracle least-square (LS) estimator is minimized. Further, we propose our design in a more general framework (single-- as well as multi--terminal settings) where observations before compression/encoding are scaled and subject to noise which is often the case in practice. Also, our optimization approach is different with those of \cite{08:Jin,07:Schizas} in the sense that we deal with sparse-structured sources, and formulate an objective function which takes into account the sparsity pattern of the source. Moreover, while the works \cite{08:Jin,07:Schizas} consider linear estimation for source reconstruction, we mainly deal with non-linear estimation for sparse source reconstruction.


\subsection{Organization} \label{sec:org}
The rest of the paper is organized as follows. In \secref{sec:problem}, we describe the single--terminal system model, and provide some preliminary analysis. Our optimization procedure for the single-terminal scenario is proposed in \secref{sec:design}, and closed-form solutions to the optimization procedure in some special cases are derived in \secref{sec:special}. We study sensing matrix design in multi--terminal systems for orthogonal MAC and coherent MAC in \secref{sec:multi_sys}. We discuss computational complexity of the proposed design procedure in \secref{sec:complexity}. The performance comparison of the proposed optimization schemes with other competing methods are made in \secref{sec:sim}, and conclusions are drawn in \secref{sec:conclusions}. All proofs are relegated  to the Appendix.

\subsection{Notations} \label{sec:notation}
We will denote vectors and matrices by bold lower-case and upper-case letters, respectively. The cardinality of a set will be denoted by $|\cdot|$. The square identity matrix and the square all-zero matrix of dimension $n$ will be denoted by $\mathbf{I}_n$ and $\mathbf{0}_n$, respectively. The matrix operators trace and Frobenius norm will be denoted by $\mathrm{Tr}(\cdot)$, $\| \cdot \|_F$, respectively, and matrix/vector transpose by $(\cdot)^\top$. The maximum and minimum eigenvalue of a matrix are denoted by $\lambda_{\max}(\cdot)$ and $\lambda_{\min}(\cdot)$, respectively. For a vector $\mathbf{x}$ of size $n$, $\mathrm{diag}(\mathbf{x})$ denotes an $n \times n$ diagonal matrix whose diagonal elements are specified by the entries of $\mathbf{x}$. Further, $\mathrm{blkdiag}(\mathbf{X}_1,\ldots, \mathbf{X}_N)$ denotes a matrix whose diagonal blocks consist of matrices $\mathbf{X}_1,\ldots, \mathbf{X}_N$, and off-diagonal blocks are filled with zero. We will use $\mathbb{E}[\cdot]$ to denote the expectation operator. The $\ell_2$-norm of a vector $\mathbf{x}$ of size $n$ will be denoted by $\|\mathbf{x}\|_2$. The notation $\mathbf{X} \succeq \mathbf{0}$ means that the matrix $\mathbf{X}$ is positive semi-definite. Also, the optimality in some sense is shown by $(\cdot)^\star$.


\section{Single--terminal System Model} \label{sec:problem}
We study the single--terminal setup shown in \figref{fig:diagram}.
\begin{figure} [!ht]
  \centering
  \psfrag{E}[][][0.75]{CS encoder}
  \psfrag{C}[][][0.75]{Channel}
   \psfrag{D}[][][0.75]{Decoder}
  \psfrag{x}[][][0.9]{$\mathbf{x}$}
  \psfrag{A}[][][0.9]{$\mathbf{A}$}
  \psfrag{y}[][][0.8]{$\mathbf{y}$}
  \psfrag{z}[][][0.8]{$\mathbf{z}$}
   \psfrag{v}[][][0.9]{$\mathbf{v}$}
  \psfrag{w}[][][0.9]{$\mathbf{w}$}
  \psfrag{H}[][][0.9]{$\mathbf{H}$}
  \psfrag{G}[][][0.9]{$\mathbf{G}$}
  \psfrag{h}[][][0.9]{$\widehat{\mathbf{x}}$}
    \includegraphics[width=9cm]{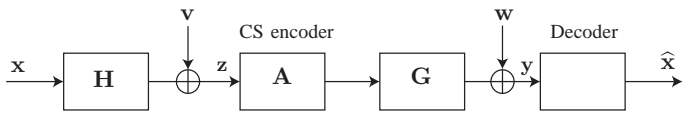}
  \caption{System model for a single--terminal system.}\label{fig:diagram}
  \centering
\end{figure}

\subsection{System Model and Key Assumptions} \label{sec:sys model}
We consider a $K$-sparse (in a known basis) vector $\mathbf{x} \in \mathbb{R}^N$ which is comprised of exactly $K$ random non-zero components ($K \ll N$). We define the support set, i.e., the locations of the non-zero components for the vector $\mathbf{x} \triangleq [x_1,\ldots,x_N]^\top$ as $\mathcal{S} \triangleq \{n \in \{1,2,\ldots,N\}: x_n \neq 0 \}$ with $|\mathcal{S}| =  K$. We assume that the non-zero components of the source vector $\mathbf{x}$ are distributed according to a Gaussian distribution $\mathcal{N}(\mathbf{0},\mathbf{R})$, where $\mathbf{R} = \mathbb{E}[\mathbf{x}_\mathcal{S} \mathbf{x}_\mathcal{S}^\top] \in \mathbb{R}^{K \times K}$ is the known covariance matrix of the $K$ non-zero components of $\mathbf{x}$, and $\mathbf{x}_\mathcal{S} \in \mathbb{R}^K$ denotes the components of $\mathbf{x}$ indexed by the support set $\mathcal{S}$. Note that the Gaussian sparse signal is compressible in nature. That is to say, the sorted amplitudes of a Gaussian sparse vector's entries, in descending order, decay fast with respect to sorted indices. Note that $\mathbf{R}$ is a positive definite matrix which is not necessarily scaled identity, i.e., the nonzero off-diagonal elements of $\mathbf{R}$ allow the non-zero components of $\mathbf{x}$ to be correlated. The elements of the support set $\mathcal{S}$ are drawn uniformly at random from the set of all ${N \choose K}$ possibilities, denoted by $\Omega$, i.e., $|\Omega|= {N \choose K}$. In other words, $p(\mathcal{S}) = 1/{N \choose K}$, where $p(\mathcal{S})$ represents the probability that a support set $\mathcal{S}$ is chosen from the set $\Omega$. The uniform distribution is chosen for simplicity of presentation, however, extensions to other types of distributions  are straightforward. We also denote the known covariance matrix of the entire sparse source vector by $\mathbf{R}_x \triangleq \mathbb{E}[\mathbf{xx}^\top] \in \mathbb{R}^{N \times N}$.

We model the uncertainty or mismatch in some physical aspect via a source-to-sensor channel described as following. The source is linearly scaled via a fixed matrix $\mathbf{H} \in \mathbb{R}^{L \times N}$ whose output is corrupted by an additive white noise $\mathbf{v} \in \mathbb{R}^L$ uncorrelated with the source, where $\mathbf{v} \sim \mathcal{N}(\mathbf{0},\sigma_v^2 \mathbf{I}_L)$. For transmission over a noisy channel, the noisy observations should be compressed and then encoded. Here, we assume that the bandwidth of the noisy observation $\mathbf{z \triangleq Hx + v} \in \mathbb{R}^L$ is compressed via a full row-rank compressed sensing transformation matrix $\mathbf{A} \in \mathbb{R}^{M \times L}$, where $M < L$. We also assume that $M < N$. The compressed measurements are simultaneously encoded under an available average transmit power constraint, and then transmitted over a channel, represented by a fixed channel matrix $\mathbf{G} \in \mathbb{R}^{M \times M}$ and additive white noise $\mathbf{w} \in \mathbb{R}^M$. We assume that the channel matrix is given by $\mathbf{G} = g \mathbf{I}_M$, and we let the additive channel noise be distributed as $\mathbf{w} \sim \mathcal{N}(\mathbf{0},\sigma_w^2 \mathbf{I}_M)$, which is uncorrelated with the source $\mathbf{x}$ and source-to-sensor noise $\mathbf{v}$. The rationale behind the scaled identity assumption of the channel matrix is that there is no inter-symbol interference between message transmissions over the communication link, and the channel is assumed to remain constant during each observation period \cite{08:Jin}. This technical assumption also makes our design procedure tractable. Now, the received vector at the decoder becomes
\begin{equation} \label{eq:measurement}
\begin{aligned}
	\mathbf{y} &= \mathbf{GAz + w} = g\mathbf{AH x} + \underbrace{g\mathbf{Av + w}}_{\triangleq \mathbf{n}}.& 
\end{aligned}
\end{equation}
Denoting the total noise in the system by $\mathbf{n} \triangleq g\mathbf{Av + w} \in \mathbb{R}^M$, then the covariance matrix associated with the total noise $\mathbf{n}$, denoted by $\mathbf{R}_n \in \mathbb{R}^{M \times M}$, can be calculated as 
\begin{equation} \label{eq:cov nosie pre}
\begin{aligned}
	\mathbf{R}_n \triangleq \mathbb{E}\{\mathbf{nn}^\top\} = g^2 \sigma_v^2 \mathbf{AA}^\top + \sigma_w^2 \mathbf{I}_M.
\end{aligned}
\end{equation}

Finally, at the receiving-end, the decoder which is characterized by a (potentially non-linear) mapping $\mathbb{R}^M \rightarrow \mathbb{R}^N$ provides the estimate of the source from corrupted measurements. We discuss the functionality of the decoder next. 

\subsection{Developing MMSE Estimation} \label{sec:perf}

Based on the aforementioned assumptions in \secref{sec:sys model}, it would be possible (see e.g. \cite{09:Elad}) to find a closed-form expression for the MMSE estimation of the source given the received signal vector $\mathbf{y}$. The MMSE estimator, denoted by $\widehat{\mathbf{x}}^\star \in \mathbb{R}^N$, minimizes the MSE by definition, and inherits the following structure (see e.g. \cite{09:Elad,10:Protter})
\begin{equation} \label{eq:struct MMSE}
	\widehat{\mathbf{x}}^\star = \sum_{\mathcal{S} \subset \Omega} \beta (\mathcal{S},\mathbf{y}) \mathbb{E}[\mathbf{x}| \mathbf{y},\mathcal{S}],
\end{equation}
where $\Omega$ represents the set of all ${N \choose K}$ support set possibilities, and $\beta(\mathcal{S},\mathbf{y})$'s are the weighting coefficients (non-linear in $\mathbf{y}$) such that $\sum_\mathcal{S} \beta(\mathcal{S},\mathbf{y}) = 1$. Further, $\mathbb{E}[\mathbf{x}| \mathbf{y},\mathcal{S}] \in \mathbb{R}^N$ is the conditional mean of the source given a possible support set $\mathcal{S}$ and observation $\mathbf{y}$. The conditional mean in \eqref{eq:struct MMSE} given a possible support set $\mathcal{S}$ can be expressed as (see  \cite{09:Elad}) $\mathbb{E}[\mathbf{x}| \mathbf{y},\mathcal{S}] = $
\begin{equation} \label{eq:oracle MMSE est}
	g \left(\mathbf{R}^{-1} \! + g^2   \left(\mathbf{H}^\top \mathbf{A}^\top \right)_\mathcal{S} \; \mathbf{R}_n^{-1} \; \left(\mathbf{A} \mathbf{H} \right)_\mathcal{S} \right)^{-1} \left(\mathbf{H}^\top \mathbf{A}^\top \right)_\mathcal{S} \; \mathbf{R}_n^{-1} \mathbf{y},
\end{equation}
where $(\cdot)_\mathcal{S}$ denotes the columns of a matrix indexed by the support set $\mathcal{S}$, and $\mathbf{R}_n$ is shown by \eqref{eq:cov nosie pre}. The MMSE estimator \eqref{eq:struct MMSE} gives the lowest possible MSE for a sparse source in the system of \figref{fig:diagram}. However, the MSE, itself, does not have a closed-form expression, and typically it is not straightforward to optimize the sensing matrix. In such situations, stochastic  optimization \cite{03:Spall}  based on gradient estimation methods (also known as {\em simulation based optimization methods}) can be an approach to address the optimization problem. However, this is beyond the scope of the current paper. Thus, we propose an alternative sensing matrix optimization method by minimizing a lower-bound on the MSE.

\subsection{Developing a Lower-bound on MSE} \label{sec:pre_analysis}

In order to analytically tackle the sensing matrix design problem, we consider a lower-bound on the MSE, and adopt the bound as the objective for the  design optimization procedure.

We bound the MSE of the MMSE estimator by that of the \textit{oracle} MMSE estimator, i.e., an \textit{ideal} estimator which has  perfect knowledge of the support set \textit{a priori}. By definition, the oracle estimator is calculated as the conditional expectation $\widehat{\mathbf{x}}^{(or)} \triangleq \mathbb{E}[\mathbf{x}| \mathbf{y},\mathcal{S}]$, as shown in \eqref{eq:oracle MMSE est}, given \textit{a priori} known (but random) support set $\mathcal{S}$ and noisy observations $\mathbf{y}$. Notice that the conditional expectation given the support set is Gaussian distributed, resulting in  the following MSE
\begin{equation} \label{eq:oracle_MSE}
\begin{aligned}
	&\mathrm{MSE}^{(lb)} \triangleq \mathbb{E}[\| \mathbf{x} - \widehat{\mathbf{x}}^{(or)} \|_2^2] =  \mathbb{E}[\| \mathbf{x}_\mathcal{S} - \widehat{\mathbf{x}}_\mathcal{S}^{(or)} \|_2^2] &\\
	&\overset{(a)}{=} \sum_{\mathcal{S} \subset \Omega} p(\mathcal{S})  \Tr\left\{ \left(\mathbf{R}^{-1} + g^2 (\mathbf{H}^\top \mathbf{A}^\top)_\mathcal{S} \;\mathbf{R}_n^{-1} \; (\mathbf{A} \mathbf{H})_\mathcal{S} \right)^{-1} \right\},&
\end{aligned}
\end{equation}
where $(a)$ follows by averaging over all random supports sets, and the results in Bayesian estimation (see, e.g., \cite{93:Kay}). Further, $p(\mathcal{S})= 1 / {N \choose K}$ represents the probability of random selection of the support set from the set of all possibilities $\Omega$.

To be able to formulate the MSE in \eqref{eq:oracle_MSE} in terms of the sensing matrix $\mathbf{A}$, we define, as in \cite{12:Chen}, the matrix $\mathbf{E}_\mathcal{S} \in \mathbb{R}^{N \times K}$ which is formed by taking an  identity matrix of order $N \times N$ and deleting the columns indexed by the set $\mathcal{S}$. Then, we rewrite 
\begin{equation} \label{eq:oracle_MSE_2}
	\mathrm{MSE}^{(lb)} \! \!= \! \sum_{\mathcal{S}} \! \frac{1}{{N \choose K}} \! \Tr \left\{ \left(\mathbf{R}^{-1} \!+\! g^2 \mathbf{E}_\mathcal{S}^\top \mathbf{H}^\top \mathbf{A}^\top \mathbf{R}_n^{\! -1} \mathbf{A} \mathbf{H} \mathbf{E}_\mathcal{S} \right)^{\!-1} \right\}.
\end{equation}

It should be mentioned that the sparsity level $|\mathcal{S}|$ is typically estimated in practice \cite{12:Lopes}. However,  throughout this paper, it is assumed to be \textit{perfectly} known. This is, of course, a generic trend in the theory of CS due to the analytical simplicity it offers \cite{08:Candes}. For example, several important greedy-search sparse reconstruction (see, e.g., OMP \cite{07:Tropp}, Subspace Pursuit \cite[Algo. 1]{09:Dai}, CoSamp \cite[Algo. 1]{10:Needell}) and Bayesian-based sparse reconstruction techniques (see, e.g., MAP, MMSE, RandOMP \cite{09:Elad}) have been developed based on the assumption of perfect knowledge of the sparsity level. Furthermore, performance guarantee bounds of several sparse reconstruction algorithms have also been studied based on this fact \cite{10:Haim}. If the sparsity level is not exactly known, but follows some statistical behavior with a known probability density function (pdf), the formulation in \eqref{eq:oracle_MSE_2} can be extended as follows. As suggested in \cite[Chap. 11]{10:Elad_book}, as opposed to $p(|\mathcal{S}|) = \delta(|\mathcal{S}| - K)$ which is the key assumption in our studied system model, i.e., $|\mathcal{S}| = K$ with probability $1$, one might assume that $p(|\mathcal{S}|) \propto  1/|\mathcal{S}|$ or $p(|\mathcal{S}|)  \propto \exp(-|\mathcal{S}|)$ in order to promote sparsity, where $p(|\mathcal{S}|)$ is the probability that the size of support set is $|\mathcal{S}|$.  Under this assumption, by marginalizing over the cardinality of the support set, it follows that

{\small\begin{equation} \label{eq:modified_MSE_lb}
\begin{aligned}
	&\mathrm{MSE}^{(lb)} = & \\ 
	&\sum_{|\mathcal{S}|=1}^{K'} \! \! p(|\mathcal{S}|) \! \! \! \sum_{\mathcal{S} \subset \Omega_{|\mathcal{S}|}'} \! \! \! p(\mathcal{S} \big| |\mathcal{S}|) \Tr \left\{ \left(\mathbf{R}^{-1} \!+\! g^2 \mathbf{E}_\mathcal{S}^\top \mathbf{H}^\top \mathbf{A}^\top \mathbf{R}_n^{\! -1} \mathbf{A} \mathbf{H} \mathbf{E}_\mathcal{S} \right)^{\!-1} \right\} &
\end{aligned}
	\end{equation}}

\noindent where $1 \leq K' \leq M$ is an integer denoting an upper-bound on the sparsity level, and $\Omega_{|\mathcal{S}|}'$ is a set of all possible support sets with cardinality $|\mathcal{S}|$. Further, $p(\mathcal{S} \big| |\mathcal{S}|)$ denotes the conditional probability of selection of the support $\mathcal{S}$  given cardinality $|\mathcal{S}|$ from the set of all possibilities $\Omega'_{|\mathcal{S}|}$. Our results, developed in this paper, can be easily extended under the new formulation in \eqref{eq:modified_MSE_lb}. However, for the sake of brevity and simplicity of presentation, we will use $\mathrm{MSE}^{(lb)}$  expressed by \eqref{eq:oracle_MSE_2} for our subsequent analysis.

\subsection{Relation to Mutual Coherence} \label{sec:rel_mu_coh}
As discussed earlier, our design goal is to optimize the sensing matrix $\mathbf{A}$ with respect to minimizing $\mathrm{MSE}^{(lb)}$. It should be mentioned that given a sensing matrix, the task of the decoder is to \textit{estimate} the sparse source with high accuracy by employing sparse reconstruction algorithms. For this purpose, sparse reconstruction algorithms need to \textit{detect} the support set precisely. Precision in support detection and accuracy in estimation of sparse reconstruction algorithms are typically determined with the help of mutual coherence $\mu$, shown by \eqref{eq:mutual co}. Let us denote by $\mathbf{S} \in \mathbb{R}^{N \times N}$ a diagonal matrix which makes the columns of the matrix $\mathbf{A}$ normalized to unit $\ell_2$-norm. This is done using the transformation $\widetilde{\mathbf{A}} = \mathbf{A} \mathbf{S}$, where $\widetilde{\mathbf{A}}$ is a sensing matrix with normalized columns. We also note that both matrices $\mathbf{A}$ and $\widetilde{\mathbf{A}}$ have the same mutual coherence.

In the following, we show a relation between $\mathrm{MSE}^{(lb)}$ and $\mu$ through a lower-bound and an upper-bound. We use a simplified measurement equation by assuming that $\mathbf{v=0}$ and $\mathbf{H} = \mathbf{I}_N$ in (\ref{eq:measurement}), that yields $\mathrm{MSE}^{(lb)} =  \sum_{\mathcal{S}} \! \frac{1}{{N \choose K}} \! \Tr \left\{ \left(\mathbf{R}^{-1} \!+\! \frac{g^2}{\sigma_w^2}  (\widetilde{\mathbf{A}} \mathbf{S}^{-1})_\mathcal{S}^\top (\widetilde{\mathbf{A}} \mathbf{S}^{-1})_\mathcal{S} \right)^{\!-1} \right\}$. We denote by $s_1$ and $s_2$ the maximum and minimum diagonal elements of $\mathbf{S}^{-1}$, respectively, then by the Gershgorin disc theorem, all the eigenvalues of $ \widetilde{\mathbf{A}}_\mathcal{S}^\top \widetilde{\mathbf{A}}_\mathcal{S}$ lie in the range $[1\!-\!(K-1)\mu,1\!+\!(K+1)\mu]$  \cite[Chapter 5.2.3]{10:Elad_book}, where it follows, using mathematical manipulations, that

\begin{equation} \label{eq:relation_mu}
\frac{K}{\lambda_{\max}(\mathbf{R}^{-1}) \!+\! \frac{g^2 s_1}{\sigma_w^2} (1+K\mu)} \leq \mathrm{MSE}^{(lb)} \! \leq \! \frac{K}{\lambda_{\min}(\mathbf{R}^{-1}) \!+\! \frac{g^2 s_2}{\sigma_w^2} (1\!-\!K \mu)}.
\end{equation}

Notice that the bounds in \eqref{eq:relation_mu} become tight when $\mu$ is small, and loose when it is large. In order to shed some light into the meaning of \eqref{eq:relation_mu}, we show the following example.

\begin{ex}
In this example, we show a comparison between the Gaussian sensing matrix (a standard approach in generating sensing matrices), and our proposed sensing matrix design (described in details in the subsequent sections) which is based on minimization of $\mathrm{MSE}^{(lb)}$. Basically, we numerically demonstrate how the proposed design affects the mutual coherence of a sensing matrix. The comparison, reported in \figref{fig:MSE_mu}(a), is demonstrated in terms of mutual coherence $\mu$ and number of measurements $M$. We set $N = 48$ and $K=2$, and vary $M$ from 12 to 48 in a step size 4. The covariance matrix $\mathbf{R}$ is generated according to the exponential model with correlation $\rho=0.5$ (cf. \secref{sec:setup}). Further, $g^2/\sigma_w^2 = 25$, and the total power (shown later by \eqref{eq:power}) is fixed at $P = 10$ dB. As can be seen, $\mu$ decreases by increasing $M$ since the sensing matrix behaves similar to an orthogonal transform. Moreover, the proposed design, which aims at minimizing $\mathrm{MSE}^{(lb)}$, provides a lower $\mu$ than that of the Gaussian sensing matrix. The efficiency of the proposed sensing matrix in lowering the mutual coherence can be seen from another angle by interpreting the bounds in \eqref{eq:relation_mu}. In \figref{fig:MSE_mu}(b), we plot the upper- and lower-bounds  in \eqref{eq:relation_mu}, as well as the value of $\mathrm{MSE}^{(lb)}$. We observe that when the number of measurements are sufficiently large for a sensing matrix to have a small $\mathrm{MSE}^{(lb)}$, then the upper- and lower-bounds become tight, i.e., $\mu$ becomes small. Thus, in this regime, since the proposed design is based on minimization of $\mathrm{MSE}^{(lb)}$, the optimized sensing matrix has a smaller $\mu$ compared to other types of sensing matrices. Note that, as mentioned earlier, a smaller $\mu$ generally improves the performance of sparse reconstruction algorithms in terms of, e.g., sparse reconstruction accuracy, support set detection, etc. In our numerical studies, later in \secref{sec:sim}, we will show how the proposed design will improve MSE performance as well as probability of support set recovery via numerical studies. 
A rigorous and general analysis of probability of support recovery with our proposed sensing matrix design and a specific sparse reconstruction algorithm is clearly difficult and will be pursued in future work.

\begin{figure} [!ht]
\hspace{-1.1cm}
    \includegraphics[width=1.2\columnwidth,height=4.5cm]{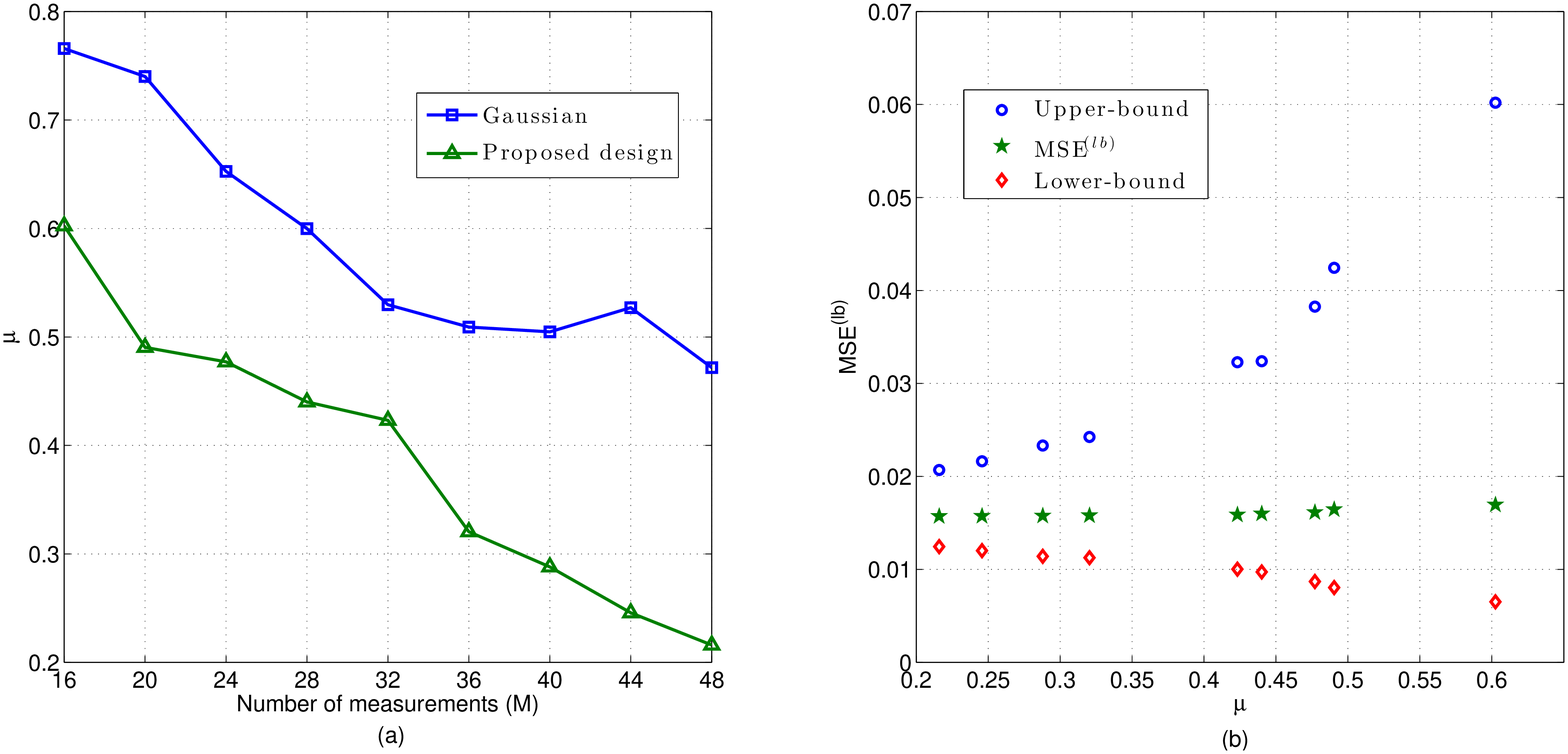}
  \caption{(a) A comparison between Gaussian sensing matrix and proposed sensing matrix design in terms of mutual coherence $\mu$ and number of measurements $M$. (b) The lower-bound and upper-bound on $\mathrm{MSE}^{(lb)}$. The lowest (or largest) $\mu$ corresponds to $M = 48$ (or $12$). }\label{fig:MSE_mu}
\end{figure}
\end{ex}
\section{Design Methodology for Single--terminal Case}  \label{sec:design}
In this section, we offer a design procedure for optimization of the sensing matrix $\mathbf{A}$ with the objective of minimizing the lower-bound \eqref{eq:oracle_MSE_2}. The optimization is performed at the decoder, and we assume that the decoder knows the sensor observation models and the source-to-sensor and sensor-to-decoder channels. 

We assume that the bandwidth is constrained, i.e., we have $M < N$ number of observations. Further, let $P$ be total available power, then the average transmit power constraint can be written as
\begin{equation} \label{eq:power}
\begin{aligned}
	\mathbb{E}[\|\mathbf{AHx + Av} \|_2^2] &= \mathbb{E}[ \Tr \{(\mathbf{AHx + Av})(\mathbf{AHx + Av})^\top\} ]& \\
	&= \Tr \{\mathbf{AH} \mathbf{R}_x \mathbf{H}^\top \mathbf{A}^\top + \sigma_v^2 \mathbf{AA}^\top\} \leq P.&
\end{aligned}
\end{equation}

Minimizing the lower-bound \eqref{eq:oracle_MSE_2} subject to the average power constraint \eqref{eq:power} yields 
\begin{equation} \label{eq:opt 1}
\begin{aligned}
	&\underset{\mathbf{A}}{\text{minimize}} \hspace{0.25cm} \mathrm{MSE}^{(lb)}& \\
	& \text{subject to} \hspace{0.25cm} \Tr \{\mathbf{A} (\mathbf{HR}_x \mathbf{H}^\top + \sigma_v^2 \mathbf{I}_N) \mathbf{A}^\top\} \leq P .& 
\end{aligned}
\end{equation}
The optimal solution of the optimization problem \eqref{eq:opt 1} is equivalent to that of the optimization problem given by the following theorem.
  
\begin{theorem} \label{theo:sing_ter}
	Let $\mathbf{Q} \triangleq \mathbf{A}^\top \mathbf{A} \in \mathbb{R}^{L \times L}$, then the optimization problem \eqref{eq:opt 1} can be equivalently solved by 

{\small \begin{equation} \label{eq:opt 1_final}
\begin{aligned}
	&\underset{\mathbf{Q},\mathbf{X}_\mathcal{S},\mathbf{Y}}{\text{minimize}} \hspace{0.25cm} \sum_\mathcal{S} \Tr \{\mathbf{X}_\mathcal{S}\}	& \\
	&\text{subject to} \hspace{0.25cm} \left[
\begin{array}{c c}
   \mathbf{R}^{-1} + \frac{g^2}{\sigma_w^2} \mathbf{D}_\mathcal{S}^\top  \mathbf{Q} \mathbf{D}_\mathcal{S}  - \mathbf{D}_\mathcal{S}^\top \mathbf{Y} \mathbf{D}_\mathcal{S} & \mathbf{I}_K \\ 
  \mathbf{I}_K  &   \mathbf{X}_\mathcal{S}   \\
\end{array}
\right] \succeq \mathbf{0} &\\
	&\hspace{1.5cm}  \left[
	\begin{array}{c c}
	   \mathbf{Y} & \frac{g}{\sigma_w}\mathbf{Q} \\ 
	  \frac{g}{\sigma_w}  \mathbf{Q}  &  \frac{\sigma_w^2}{g^2 \sigma_v^2} \mathbf{I}_L  + \mathbf{Q} \\
	\end{array}
	\right] \succeq \mathbf{0} , \; \forall  \mathcal{S}  & \\
	 &\hspace{1.5cm} \Tr \{(\mathbf{HR}_x \mathbf{H}^\top \!+\! \sigma_v^2 \mathbf{I}_L)  \mathbf{Q}\} \leq P ,\hspace{0.1cm} \mathbf{Q} \succeq \mathbf{0}, \hspace{0.1cm} \mathrm{rank}(\mathbf{Q}) \!=\! M,&
\end{aligned}
\end{equation}}
where $\mathbf{D}_\mathcal{S} \triangleq \mathbf{H} \mathbf{E}_\mathcal{S}$, and the matrices  $\mathbf{Q}$, $\mathbf{X}_\mathcal{S} \in \mathbb{R}^{K \times K}$ and $\mathbf{Y} \in \mathbb{R}^{L \times L}$ are optimization variables. 
\end{theorem}

\begin{remark}
	The last two constraints in \eqref{eq:opt 1_final} appear due to the variable transformation $\mathbf{Q} = \mathbf{A}^\top \mathbf{A}$ which is a rank-$M$ positive semi-definite matrix. The difficulty of \eqref{eq:opt 1_final} is due to the rank constraint which makes the optimization problem non-convex in general. However, the constraint can be relaxed, and the remaining problem becomes convex -- a technique known as semi-definite relaxation (SDR). Note that the optimal value of the SDR problem can  only be used to give a lower-bound on the optimal cost of the original problem.
\end{remark}

Next, we develop a three-stage optimization procedure, shown in \procref{proc}, in order to approximately solve for $\mathbf{A}$ in the non-convex optimization problem \eqref{eq:opt 1_final}. 
\begin{algorithm}
\caption{Three-stage optimization procedure for solving \eqref{eq:opt 1_final}}\label{proc}
\begin{algorithmic}[1]
\STATE{\textbf{input:} measurement vector: $\mathbf{y}$, sparsity level $K$, covariance matrices $\mathbf{R}_x$ and $\mathbf{R}$, channel gain $g$ and noise variances $\sigma_v^2$ and $\sigma_w^2$.}
\STATE{\textbf{Semi-definite relaxation (SDR): } Solve \eqref{eq:opt 1_final} by dropping the rank constraint for the optimal $\mathbf{Q}^\star$.}
\STATE{\textbf{Low-rank reconstruction: } Solve
\begin{equation} \label{eq:opt_rec_A_appx}
	\begin{aligned}
		&\mathbf{A}^\star = \underset{\mathbf{A}}{{\text{arg min}}} \hspace{0.25cm} \|   \mathbf{A}^\top \mathbf{A} - \mathbf{Q}^\star \|_F^2.&
	\end{aligned}
\end{equation}
\STATE{\textbf{Power-rescaling: } Scale $\mathbf{A}^\star$ to satisfy the power constraint by equality.}}
\end{algorithmic}
\end{algorithm}

The following remarks can be considered for implementation of \procref{proc}. 
\begin{itemize}
\item  The SDR problem in step (2) is convex in $\mathbf{Q}$, and can be solved using, for example, the interior point method \cite{04:Boyd_book}. Further, in some cases, closed-form solutions exist which we discuss later in the next section. 

\item Step (3) gives an approximate solution to the sensing matrix design problem. It can be shown that the optimal $\mathbf{A}^\star$ (with respect to \eqref{eq:opt_rec_A_appx}) has a closed-form solution. Let the eigen-value decomposition (EVD) of $\mathbf{Q}^\star$ be 
\begin{equation} \label{eq:evd_Q}
	\mathbf{Q}^\star = \mathbf{U}_q \mathbf{\Gamma}_q \mathbf{U}_q^\top,
\end{equation}
where $\mathbf{\Gamma}_q = \mathrm{diag}(\gamma_{q_1},\ldots,\gamma_{q_N})$, with $\gamma_{q_1} \geq \ldots, \geq \gamma_{q_N}$, and $\mathbf{U}_q \in \mathbb{R}^{L \times L}$ is a unitary matrix whose columns are the eigen-vectors associated with the eigen-values of $\mathbf{Q}^\star$. Then,  $\mathbf{A}^\star$ has the following structure \cite{11:Kokiopoulou}
\begin{equation} \label{eq:recstr A single}
	\mathbf{A}^\star = \mathbf{U}_a \left[\mathrm{diag}(\sqrt{\gamma_{q_1}} , \ldots , \sqrt{\gamma_{q_M}}) \;\; \mathbf{0}_{M \times (L-M)}\right] \mathbf{U}_q^\top,
\end{equation}
where $\mathbf{U}_a \in \mathbb{R}^{M�\times M}$ is an arbitrary unitary matrix. 

\item We note that the resulting $\mathbf{A}^\star$ does not generally satisfy the power constraint by equality since the eigen-values $\gamma_{q_{M+1}},\ldots,\gamma_{q_N}$ are dropped in \eqref{eq:recstr A single}. Therefore, in step (4) of \procref{proc}, we rescale the resulting $\mathbf{A}^\star$ by the constant $\sqrt{P / \Tr\{(\mathbf{H} \mathbf{R}_x \mathbf{H}^\top+ \sigma_v^2 \mathbf{I}_L)\mathbf{A}^{\star \top} \mathbf{A}^{\star}\}}$ in order to satisfy the power constraint by equality. 
\end{itemize}

\begin{ex}
In order to offer insights into the effect of the rank constraint (in the optimization problem \eqref{eq:opt 1_final}) on the performance, we illustrate, in \figref{fig:illustration}, the value of the lower-bound $\mathrm{MSE}^{(lb)}$ in \eqref{eq:oracle_MSE_2} as a function of number of measurements $M$ by comparing three methods. In the \textit{first ideal} method, labeled by `full-rank optimization', we only solve the SDR problem in step (2) of \procref{proc}, and evaluate the value of  $\mathrm{MSE}^{(lb)}$. Therefore, the optimization variable $\mathbf{Q}$ is \textit{ideally} assumed to be full rank, and the value of $\mathrm{MSE}^{(lb)}$ using the resulting SDR gives a lower-bound on the optimal cost provided the rank constraint is applied. In the \textit{second} method, labeled by `rank-constrained optimization (\procref{proc})', we exploit the proposed \procref{proc}, where rank constraint is taken into consideration. In the \textit{third} method, we use the randomization technique \cite{10:Zhi} instead of step (3) in \procref{proc} which is labeled by `rank-constrained optimization (randomization)'. More precisely, using this method, we assume that the resulting sensing matrix is given by $\mathbf{A} = \mathbf{V} \mathbf{\Gamma}^{1/2} \mathbf{U}_q^\top$, where $\mathbf{V} \in \mathbb{R}^{M \times L}$ is a random matrix whose element $[\mathbf{V}]_{ij}$ is drawn from $\mathcal{N}(0, 1/ \sqrt{M})$ such that $\mathbb{E}[\mathbf{A}^\top \mathbf{A}] = \mathbf{Q}$. Note that we rescale each realization of $\mathbf{A}$ to meet the power constraint, and choose the one which gives the lowest $\mathrm{MSE}^{(lb)}$. 

In this illustration, we assume that $\mathbf{H} = \mathbf{I}_N$ and $\mathbf{v=0}$, and use the parameters: $N = 24, K = 3, \sigma_w = 0.1, g = 0.5, P = 10 \text{ dB}$, and $\rho = 0.5$ (i.e., correlation coefficient, see later in \secref{sec:setup}.). Further, in the third method, we use $1000$ randomizations. 

It is observed that the proposed method (i.e., \procref{proc}) provides a lower MSE than the randomization technique. Moreover, \procref{proc} has a lower-complexity in step (3) since the randomization technique compares all possible values of $\mathrm{MSE}^{(lb)}$ due to the random realizations of the sensing matrix. The gap between the curves labeled by `rank-constrained optimization (\procref{proc})' and `full-rank optimization', which is not a large margin, shows the loss due to imposing the rank constraint. As can be seen the loss reduces as $M$ increases. One reason is that the approximation of the sensing matrix $\mathbf{A}$ from the variable $\mathbf{Q}$ in the optimization problem \eqref{eq:opt_rec_A_appx} becomes more accurate. As a final remark, we note that if the optimization problem \eqref{eq:opt 1_final} with the rank constraint is exactly solved using some technique, then the minimum cost would lie between these two curves.  

\begin{figure}  
  \begin{center}
  \includegraphics[width=0.9\columnwidth,height=7cm]{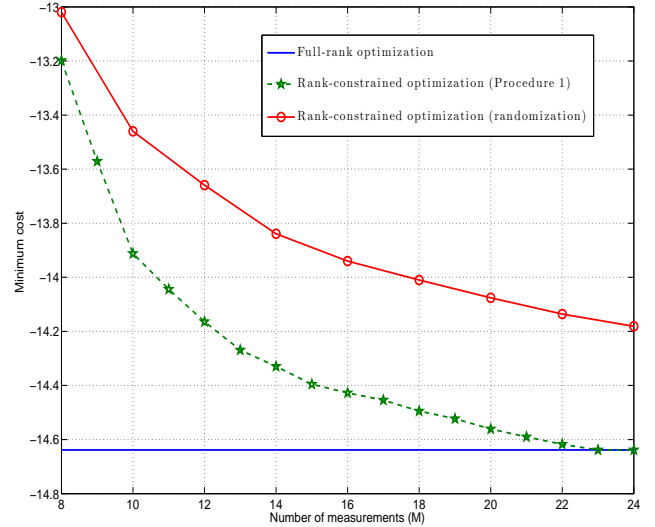}\\
  \caption{A comparison between the minimum cost of the objective in \eqref{eq:opt 1_final} with and without rank constraint.}
  \label{fig:illustration}
  \end{center}
\end{figure}

\end{ex}

\section{Special Cases} \label{sec:special}
Here, we investigate the optimization problem \eqref{eq:opt 1_final} and \procref{proc} for several special cases. 

\subsection{Special Case I ($\mathbf{R} = \sigma_x^2 \mathbf{I}_K$, $\mathbf{H} = \mathbf{I}_N$)} \label{sec:special case1}

Here, the motivation is to study a scenario where the non-zero components of the sparse source are uncorrelated, i.e., $\mathbf{R} = \sigma_x^2 \mathbf{I}_K$ and the observations before encoding are only subject to additive noise, i.e., $\mathbf{H} = \mathbf{I}_N$. Under these conditions, we have the following result.
\begin{proposition} \label{prop:sepc_1}
	Let $\mathbf{R} = \sigma_x^2 \mathbf{I}_K$ and $\mathbf{H} = \mathbf{I}_N$, then the solution to \procref{proc} is given by
	\begin{equation} \label{eq:closed_spec1_4_proof}
		\mathbf{A}^\star = \sqrt{\frac{KP}{M(\sigma_x^2 + K \sigma_v^2)}} \; \mathbf{U}_a [\mathbf{I}_M \; \; \mathbf{0}_{M \times 		(N-M)}] \mathbf{V}_a^\top, 
	\end{equation}
	where $\mathbf{U}_a \in \mathbb{R}^{M \times M}$ and $\mathbf{V}_a \in \mathbb{R}^{N \times N}$ are arbitrary unitary matrices. 
\end{proposition}

\begin{remark}
The scaling factor on the right-hand side in \eqref{eq:closed_spec1_4_proof} is to satisfy the power constraint. Further, the structure of the sensing matrix in \eqref{eq:closed_spec1_4_proof} is normally referred to as \textit{`tight frame'} \cite{08:Kovac}, which is easy to construct, and plays important roles in signal processing, denoising, coding, etc. Such structure is also optimal in certain cases, for example, the optimality of a tight frame-structured sensing matrix has been shown in \cite{12:Chen} with respect to minimizing the LS-based oracle estimator. 
\end{remark}

\subsection{Special Case II ($\mathbf{R} = \sigma_x^2 \mathbf{I}_K$, $\mathbf{v} = \mathbf{0}, \mathbf{H}:$ square full rank)} \label{sec:special case1}
Following the assumptions in this case, we have the proposition below.
\begin{proposition} \label{prop:sepc_2}
	Let $\mathbf{R} = \sigma_x^2 \mathbf{I}_K$ and $\mathbf{v} = \mathbf{0}$, and consider that $\mathbf{H}$ is a square full-rank matrix such that its SVD can be written as $\mathbf{H} = \mathbf{U}_h \mathbf{\Gamma}_h  \mathbf{V}_h^\top$, where $ \mathbf{U}_h$ and  $\mathbf{V}_h$ are $N \times N$ unitary matrices and $\mathbf{\Gamma}_h = \mathrm{diag}(\gamma_{h_1},\gamma_{h_2}, \ldots, \gamma_{h_N} )$ is a diagonal matrix containing singular values $\gamma_{h_1}<\gamma_{h_2}< \ldots < \gamma_{h_N}$. Then, the solution to \procref{proc} is given by
	\begin{equation} \label{eq:closed_spec2_4}
		\mathbf{A}^\star = \sqrt{\frac{KP}{M\sigma_x^2 }} \; \mathbf{U}_a [\mathbf{\Gamma}_a \; \; \mathbf{0}_{M \times (N-M)}] 		\mathbf{U}_h^\top, 
	\end{equation}
	where $\mathbf{U}_a \in \mathbb{R}^{M \times M}$ is an arbitrary unitary matrix, and $\mathbf{\Gamma}_a = \mathrm{diag}(\gamma_{h_1}^{-1},\ldots,\gamma_{h_M}^{-1})$. 
\end{proposition}

\begin{remark}
The scaling factor on the right-hand side in \eqref{eq:closed_spec2_4} is to satisfy the power constraint. According to \eqref{eq:closed_spec2_4} in \proref{prop:sepc_2}, the effective received measurement matrix at the decoder, i.e., $g \mathbf{AH}$, has a tight-frame structure. Interestingly, it can be also shown (see e.g. \cite{11:Duarte}) that the optimized sensing matrix derived in \eqref{eq:closed_spec2_4}, without the scaling factor, coincides with the optimal solution to the optimization problem
	\begin{equation*}
	\begin{aligned}
		&\underset{\mathbf{A}}{{\text{minimize}}} \hspace{0.25cm} \| \mathbf{H}^	\top \mathbf{A}^\top \mathbf{A} \mathbf{H}  - \mathbf{I}_N \|_F ,& \\
	\end{aligned}
	\end{equation*}
which belongs to the second category of sensing matrix design problems introduced in \secref{sec:background}. Therefore, the proposed design is capable of reducing the mutual coherence of the effective measurement matrix which, in general, improves the performance of sparse reconstruction algorithms. Also, notice that the optimal sensing matrix in \eqref{eq:closed_spec2_4} (without the power scaling factor) is the closest design -- in the Frobenius distance -- to the identity transform.
\end{remark}

\subsection{Special Case III ($\mathbf{w= 0} $, $\mathbf{H} = \mathbf{I}_N$, $\mathbf{R} = \sigma_x^2 \mathbf{I}_K$)} \label{sec:special case3}
Here, we investigate a case where the additive channel noise in the system is negligible, i.e., $\mathbf{w=0}$,  the observations before encoding are only subject to additive noise, i.e., $\mathbf{H} = \mathbf{I}_N$, and the non-zero components of the sparse source vector are uncorrelated, i.e., $\mathbf{R} = \sigma_x^2 \mathbf{I}_K$. In this case, the optimal sensing matrix to the original problem \eqref{eq:opt 1} can be derived which is given by the following proposition. 
\begin{proposition} \label{prop:sepc_3}
	Let $\mathbf{w= 0} $, $\mathbf{H} = \mathbf{I}_N$, $\mathbf{R} = \sigma_x^2 \mathbf{I}_K$. Then, the solution to the optimization problem \eqref{eq:opt 1} is given by
	\begin{equation} \label{eq:spec_case_3}
	\mathbf{A}^\star = \sqrt{\frac{KP}{M( \sigma_x^2 + K\sigma_v^2)}} \mathbf{U}_a [\mathbf{I}_{M} \: \: \mathbf{0}_{M \times (N-M)}],
\end{equation}
where $\mathbf{U}_a \in \mathbb{R}^{M \times M}$ is an arbitrary unitary matrix.  
\end{proposition}

\begin{remark}
The scaling factor on the right-hand side in \eqref{eq:spec_case_3} is to satisfy the power constraint by equality. From the result of \proref{prop:sepc_3}, as well as that of \proref{prop:sepc_1} and \proref{prop:sepc_2}, it can be observed that as long as the source is uncorrelated and the source-to-sensor channel has a special structure (identity or full-rank), then the optimized sensing matrix does not depend on the channel gain and additive noise. It should be noted, however, that the value of MSE still depends on the channel parameters.  
\end{remark}

\subsection{Special Case IV ($\mathbf{v= 0} $, $\frac{g^2}{\sigma_w^2} \rightarrow 0$)} \label{sec:special case4}
Now, we consider an asymptotic case, where the communication channel is in a noisy regime such that the ratio between the power of channel gain over the power of additive channel noise tends to zero, i.e., $g^2/\sigma_w^2 \rightarrow 0$. 
\begin{proposition} \label{prop:sepc_4}
	Let $\mathbf{v= 0} $ and $\frac{g^2}{\sigma_w^2} \rightarrow 0$, and define $\mathbf{T} \triangleq \sum_\mathcal{S}  \mathbf{D}_\mathcal{S} \mathbf{R}^2 \mathbf{D}_\mathcal{S}^\top$ and $\mathbf{Z} \triangleq \mathbf{T}^{-1/2} \mathbf{HR}_x \mathbf{H}^\top \mathbf{T}^{-1/2}$ which has the EVD $\mathbf{Z} = \mathbf{U}_z \mathbf{\Gamma}_z \mathbf{U}_z^\top$. Then, the approximate solution to \procref{proc} is asymptotically given by
	\begin{equation} \label{eq:sol asymp2}
		\mathbf{A}^\star =  \mathbf{U}_a \left[\mathrm{diag}  \left(\sqrt{\gamma_q},0,\ldots,0 \right)\; \; \mathbf{0}_{M \times (L-M)} 		\right] \mathbf{U}_q^\top,
	\end{equation}
	where $\mathbf{U}_a \in \mathbb{R}^{M \times M}$ is an arbitrary unitary matrix,  and $\gamma_q$ is the only non-zero eignevlaue of 
	\begin{equation} \label{eq:sol asymp}
		\mathbf{Q}^\star =  \mathbf{T}^{-1/2}\mathbf{U}_z \mathrm{diag} \left(\frac{P}{\gamma_{z_1}},0,\ldots,0 \right)\mathbf{U}_z^\top \mathbf{T}^{-1/2}.
	\end{equation}
	Further, $\mathbf{U}_q$ is the eigen-vector associated with the EVD of $\mathbf{Q}^\star$, and $\gamma_{z_1}$ is the smallest eigen-value of $\mathbf{Z}$.
	\end{proposition}

\begin{remark}
From \eqref{eq:sol asymp2}, it can be observed if channel condition degrades, as $g^2 / \sigma_w^2 \rightarrow 0$, the sensing matrix has only one active singular-value. \end{remark}

Up to this point, we have investigated the design of sensing matrix for the single-terminal scenario. The techniques presented so far will help us  analyze and  design sensing matrices for multi-terminal scenarios with orthogonal and coherent MAC which are described in the next section.

\section{Design Methodology for Multi--terminal Case} \label{sec:multi_sys}
In this section, we study sensing matrix design for a multi-terminal system consisting of orthogonal and coherent MAC. In orthogonal MAC, the sensors are scheduled orthogonally in time or frequency where coordination between the sensors are required, whereas in coherent MAC, all sensor transmissions occur simultaneously but require distributed phase synchrozination, also known as distributed beamforming at the sensor transmitters. Throughout the design, for both cases, we assume that the fusion center (FC) knows the sensor observation models and the source-to-sensor as well as sensor-to-decoder channels. It should be also mentioned that the optimized sensing matrix design is performed at the FC.

\subsection{Orthogonal MAC} \label{sec:orthog_mac}

We consider the following multi--terminal setup with orthogonal MAC which is shown in \figref{fig:diagram_dist}. 
\begin{figure} [!ht]
  \centering
  \psfrag{x}{$\mathbf{x}$}
  \psfrag{x}[][][0.75]{$\mathbf{x}$}
  \psfrag{O}[][][0.75]{$\mathbf{H}_1$}
   \psfrag{P}[][][0.75]{$\mathbf{H}_2$}
  \psfrag{A}[][][0.75]{$\mathbf{v}_1$}
    \psfrag{B}[][][0.75]{$\mathbf{v}_2$}
  \psfrag{L}[][][0.75]{$\mathbf{z}_1$}
  \psfrag{M}[][][0.75]{$\mathbf{z}_2$}
   \psfrag{Q}[][][0.75]{$\mathbf{A}_1$}
  \psfrag{R}[][][0.75]{$\mathbf{A}_2$}
  \psfrag{F}[][][0.75]{$\mathbf{G}_1$}
    \psfrag{G}[][][0.75]{$\mathbf{G}_2$}
  \psfrag{c}[][][0.75]{$\mathbf{w}_1$}
    \psfrag{d}[][][0.75]{$\mathbf{w}_2$}
    \psfrag{u}[][][0.75]{$\mathbf{y}_1$}
    \psfrag{v}[][][0.75]{$\mathbf{y}_2$}
  \psfrag{h}[][][0.75]{$\widehat{\mathbf{x}}$}
  \psfrag{E}[][][0.75]{CS encoder}
  \psfrag{C}[][][0.75]{Channel}
  \psfrag{Z}[][][0.75]{FC}
   \psfrag{D}[][][0.75]{Decoder}
    \includegraphics[width=9cm]{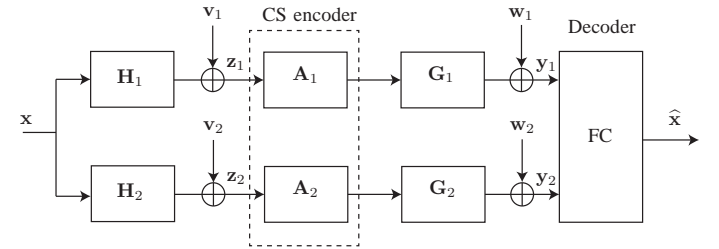}
  \caption{System model for the multi-terminal scenario with orthogonal MAC.}\label{fig:diagram_dist}
  \centering
\end{figure}

We consider the sparse source vector with the same properties as those described in \secref{sec:sys model}. Without loss of generality, we assume that the source is linearly scaled via two fixed matrices $\mathbf{H}_l \in \mathbb{R}^{L_l \times N}$ ($l \in \{1,2\}$) whose outputs are corrupted by additive noise vectors $\mathbf{v}_l$ uncorrelated with the source, where $\mathbf{v}_l \sim \mathcal{N}(0,\sigma_{v_l}^2 \mathbf{I}_{L_l})$. For transmission purposes, we suppose that the bandwidth of the noisy observations $\mathbf{z}_l \triangleq \mathbf{H}_l \mathbf{x} + \mathbf{v}_l \in \mathbb{R}^L$ is linearly compressed via the full row-rank matrix $\mathbf{A}_l \in \mathbb{R}^{M_l \times L_l}$, where $M_l < L_l$.  The compressed measurements are simultaneously encoded based on a limited power constraint budget, and then transmitted over noisy channels, represented by fixed channel matrices $\mathbf{G}_l = g_l \mathbf{I}_{M_l}$ and additive noise $\mathbf{w}_l \sim \mathcal{N}(0,\sigma_{w_l}^2 \mathbf{I}_{M_l})$, which is uncorrelated with $\mathbf{x}$ and $\mathbf{v}_l$.
The received measurement at the FC can be written as
\begin{equation} \label{eq:measurement_final_dist}
\begin{aligned}
	\widetilde{\mathbf{y}}
	= \widetilde{\mathbf{A}} \widetilde{\mathbf{H}} \mathbf{J} \mathbf{x} + \underbrace{\widetilde{\mathbf{A}} \widetilde{\mathbf{v}} + \widetilde{\mathbf{w}}}_{\triangleq \widetilde{\mathbf{n}}},
\end{aligned}
\end{equation}
where
\begin{equation} \label{eq:access}
\begin{aligned}
	\widetilde{\mathbf{y}} &\triangleq [\mathbf{y}_1^\top \; \; \mathbf{y}_2^\top ]^\top ,&
	  \mathbf{J} &\triangleq [\mathbf{I}_N \; \; \mathbf{I}_N]^\top, &\\
	  \widetilde{\mathbf{H}} &\triangleq \mathrm{blkdiag}(\mathbf{H}_{1},\mathbf{H}_{2}), &
	 \widetilde{\mathbf{A}} &\triangleq \mathrm{blkdiag}(g_1\mathbf{A}_1,g_2\mathbf{A}_2)
	 , &\\ 
	 \widetilde{\mathbf{v}} &\triangleq [\mathbf{v}_1^\top \; \; \mathbf{v}_2^\top ]^\top ,&
	  \widetilde{\mathbf{w}} &\triangleq [\mathbf{w}_1^\top \; \; \mathbf{w}_2^\top ]^\top .&
\end{aligned}
\end{equation}

Denoting the total noise in the system by $\widetilde{\mathbf{n}} = \widetilde{\mathbf{A}} \widetilde{\mathbf{v}} + \widetilde{\mathbf{w}} \in \mathbb{R}^{M_1 + M_2}$,   the covariance matrix associated with the total noise, denoted by $\widetilde{\mathbf{R}}_n \in \mathbb{R}^{(M_1 + M_2) \times (M_1 + M_2)}$, is $\widetilde{\mathbf{R}}_n \triangleq \mathbb{E} [\widetilde{\mathbf{n}} \widetilde{\mathbf{n}}^\top] =$ 
\begin{equation} \label{eq:cov nosie}
\begin{aligned}
	&\! \mathrm{blkdiag}(g_1^2 \sigma_{v_1}^2  \mathbf{A}_1\mathbf{A}_1^\top \!\!+\! \sigma_{w_1}^2 \mathbf{I}_{M_1}, 
	g_2^2 \sigma_{v_2}^2 \mathbf{A}_2\mathbf{A}_2^\top \!\!+ \!\sigma_{w_2}^2 \mathbf{I}_{M_2}). &
\end{aligned}
\end{equation}

For the design of sensing matrices in the system of \figref{fig:diagram_dist}, we aim at minimizing a lower-bound on the MSE of the sparse source. Similar to the steps taken in \secref{sec:pre_analysis}, we can derive the oracle MMSE estimator. Following \eqref{eq:oracle MMSE est}, the oracle estimator of $\mathbf{x}$ given the measurements \eqref{eq:measurement_final_dist} can be written as $\mathbb{E}[\mathbf{x} | \widetilde{\mathbf{y}}, \mathcal{S}] =$
\begin{equation} \label{eq:oracle_est_dist}
		 \left(\mathbf{R}^{\! -1} \!+ \! \left(\mathbf{J}^\top \widetilde{\mathbf{H}}^\top \widetilde{\mathbf{A}}^\top \right)_{\! \mathcal{S}}  \widetilde{\mathbf{R}}_n^{\!-1} \left(\widetilde{\mathbf{A}} \widetilde{\mathbf{H}} \mathbf{J} \right)_{\! \mathcal{S}} \right)^{\! -1} \left(\widetilde{\mathbf{H}}^\top \widetilde{\mathbf{A}}^\top \right)_{\! \mathcal{S}} \; \widetilde{\mathbf{R}}_n^{\! -1} \widetilde{\mathbf{y}}.
\end{equation}
Recalling that $\mathbf{E}_\mathcal{S} \in \mathbb{R}^{N \times K}$ is formed by taking an  identity matrix of order $N \times N$ whose columns indexed by the support set $\mathcal{S}$ are deleted, the oracle estimator in \eqref{eq:oracle_est_dist} gives the oracle MSE determined as following
\begin{equation} \label{eq:oracle_MSE_dist}
\begin{aligned}
	\mathrm{MSE}_o^{(lb)}
	\! \! \! \triangleq \! \! \sum_{\mathcal{S}} \! \frac{1}{{N \choose K}} \! \Tr \! \left\{ \! \left(\mathbf{R}^{-1} \!+\! \mathbf{E}_\mathcal{S}^\top \mathbf{J}^\top \widetilde{\mathbf{H}}^\top \widetilde{\mathbf{A}}^\top \widetilde{\mathbf{R}}_n^{\! -1} \widetilde{\mathbf{A}} \widetilde{\mathbf{H}} \mathbf{J} \mathbf{E}_\mathcal{S} \right)^{\!-1} \!\right\}.
\end{aligned}
\end{equation}

So as to formulate the sensing matrix optimization problem, we determine the total average transmit power constraint as
\begin{equation} \label{eq:power_ortho}
\begin{aligned}
	&\sum_{l=1}^2 \mathbb{E}[\| \mathbf{A}_l \mathbf{H}_l  \mathbf{x} + \mathbf{A}_l \mathbf{v}_l \|_2^2] & \\
	&= \sum_{l=1}^2 \Tr \{\mathbf{A}_l\mathbf{H}_l \mathbf{R}_x \mathbf{H}_l^\top \mathbf{A}_l^\top + \sigma_{v_l}^2 \mathbf{A}_l\mathbf{A}_l^\top\} \leq P,&
\end{aligned}
\end{equation}
where $P$ is the total available power, and the last equality is obtained by straightforward mathematical manipulations. 

It should be also mentioned that, throughout the design for the multi-terminal systems, we consider that the total power for the sensors are constrained. However, our design procedure can be applied also when power per sensor is constrained. 

Now, we pose the following optimization problem
\begin{equation} \label{eq:opt 1_dist}
\begin{aligned}
	&\underset{\mathbf{A}_1, \mathbf{A}_2}{\text{minimize}} \hspace{0.25cm} \mathrm{MSE}_o^{(lb)}& \\
	& \text{subject to} \hspace{0.25cm} \sum_{l=1}^2 \Tr \{\mathbf{A}_l\mathbf{H}_l \mathbf{R}_x \mathbf{H}_l^\top \mathbf{A}_l^\top + \sigma_{v_l}^2 \mathbf{A}_l\mathbf{A}_l^\top\} \leq P ,& 
\end{aligned}
\end{equation}
where $\mathrm{MSE}_o^{(lb)}$ is shown by \eqref{eq:oracle_MSE_dist}. We have the following result.

\begin{theorem} \label{theo:mult_ter_ortho}
 The optimization problem \eqref{eq:opt 1_dist} can be equivalently solved by the problem \eqref{eq:opt 1_final_dist}, on top of next page,
 \begin{figure*}[!ht]

{\small
\setcounter{MYtempeqncnt}{\value{equation}}
\begin{equation}  \label{eq:opt 1_final_dist}
\begin{aligned}
  &\underset{\mathbf{Q}_l,\widetilde{\mathbf{X}}_\mathcal{S},\mathbf{Y}_{l}}{\text{minimize}} \hspace{0.25cm} \sum_\mathcal{S} \Tr \{\widetilde{\mathbf{X}}_\mathcal{S}\}	& \\
	&\text{subject to} \hspace{0.25cm} \left[
\begin{array}{c c}
	\mathbf{R}^{- 1} +   \widetilde{\mathbf{E}}_\mathcal{S}^\top \mathrm{blkdiag}\left(\frac{g_1^2}{\sigma_{w_1}^2} \mathbf{Q}_1, \frac{g_2^2}{\sigma_{w_2}^2} \mathbf{Q}_2 \right) \widetilde{\mathbf{E}}_\mathcal{S}    - \widetilde{\mathbf{E}}_\mathcal{S}^\top  \mathrm{blkdiag}(\mathbf{Y}_{1} , \mathbf{Y}_{2}) \widetilde{\mathbf{E}}_\mathcal{S} & \mathbf{I}_K \\ 
 \mathbf{I}_K   &   \widetilde{\mathbf{X}}_\mathcal{S}   \\
\end{array}
\right] \succeq \mathbf{0} &\\
	&\hspace{1.5cm}  \left[
\begin{array}{c c}
	   \mathbf{Y}_l & \frac{g_l}{\sigma_{w_l}}\mathbf{Q}_l \\ 
	  \frac{g_l}{\sigma_{w_l}}  \mathbf{Q}_l  &  \frac{\sigma_{w_l}^2}{g_l^2 \sigma_{v_l}^2} \mathbf{I}_{L_l}  + \mathbf{Q}_l \\
	\end{array}
	\right] \succeq \mathbf{0} , 
	\hspace{0.1cm}  \sum_{l=1}^2 \Tr \left\{\left(\mathbf{H}_l \mathbf{R}_x \mathbf{H}_l^\top + \sigma_{v_l}^2 \mathbf{I}_{L_l} \right)\mathbf{Q}_l \right\} \leq P, 
	\hspace{0.1cm} \mathbf{Q}_l \succeq \mathbf{0},  \hspace{0.2cm}  \mathrm{rank}(\mathbf{Q}_l) = M_l, \;\; \forall{l}, \;  \mathcal{S} .&\end{aligned}
\end{equation}
\setcounter{equation}{\value{MYtempeqncnt}}}
\hrulefill
\end{figure*}
\setcounter{equation}{30}
where we have defined $\widetilde{\mathbf{E}}_\mathcal{S} \triangleq \widetilde{\mathbf{H}} \mathbf{J}  \mathbf{E}_\mathcal{S}$, and further $\mathbf{Q}_l \triangleq \mathbf{A}_l^\top \mathbf{A}_l \in \mathbb{R}^{L_l \times L_l}$, $\widetilde{\mathbf{X}}_\mathcal{S} \in \mathbb{R}^{K \times K}$ and $\mathbf{Y}_l \in \mathbb{R}^{L_l \times L_l}$, $l \in \{1,2\}$, are optimization variables. 
\end{theorem}

\begin{remark} \label{rem:ortho}
	Note that the optimization problem \eqref{eq:opt 1_final_dist} is not generally convex due to the rank constraints. Similar to \procref{proc}, we give an approach in order to approximately solve for $\mathbf{A}_l$ ($l \in \{1,2\}$). Ignoring the rank constraints, the resulting SDR problem would be convex jointly in all optimization variables. Denoting the optimal solution of the SDR problem by $\mathbf{Q}_l^\star$, and taking EVD, we obtain $\mathbf{Q}_l^\star = \mathbf{U}_{q_l} \mathbf{\Gamma}_{q_l} \mathbf{U}_{q_l}^\top$, where $\mathbf{U}_{q_l} \in \mathbb{R}^{L_l \times L_l}$ is a unitary matrix, whose columns are eigen-vectors associated with the eigen-values  of the matrix $\mathbf{\Gamma}_{q_l} = \mathrm{diag} \left(\gamma_{q_{1_l}},  \ldots , \gamma_{q_{L_l}} \right)$ such that $\gamma_{q_{1_l}} \geq \ldots \geq \gamma_{q_{L_l}}$. Now, we can approximately reconstruct the rank-$M_l$ sensing matrix $\mathbf{A}_l^\star$ from $\mathbf{Q}_l^\star$ by admitting the $M_l$ largest eigen-values of $\mathbf{Q}_l^\star$, and by letting $\mathbf{A}_l^\star$ have the following structure
\begin{equation} \label{eq:recstr A ortho}
	\mathbf{A}_l^\star = \mathbf{U}_{a_l} \left[\mathrm{diag}(\sqrt{\gamma_{q_{1_l}}} , \ldots , \sqrt{\gamma_{q_{M_l}}}) \;\; \mathbf{0}_{M_l \times (L_l-M_l)}\right] \mathbf{U}_{q_l}^\top,
\end{equation}
where $\mathbf{U}_{a_l} \in \mathbb{R}^{M_l�\times M_l}$ is an arbitrary unitary matrix. 

Here, there is a slight difference in power-rescaling the matrix $\mathbf{A}_l^\star$ compared to the single-terminal case. Since each terminal is subject to different channel gains and noises, $\mathbf{A}_1^\star$ and $\mathbf{A}_2^\star$ need to be scaled differently. Therefore, we give a weighting coefficient to each sensing matrix, i.e., $\mathbf{A}_l^\star \rightarrow \sqrt{\alpha_l} \mathbf{A}_l^\star$, where $\alpha_l \geq 0$ is the weighting coefficient to be optimized, and $\mathbf{A}_l^\star$ is already determined from the previous stage. Then, we solve the optimization problem \eqref{eq:opt 1_final_dist} with new optimization variables $\alpha_1 \geq 0$ and $\alpha_2 \geq 0$ instead of $\mathbf{Q}$ which is known at this stage. Note that the resulting optimization problem becomes convex in $\alpha_1$ and $\alpha_2$ and can be solved efficiently using any convex solver. The final rescaled optimized sensing matrices become $\sqrt{\alpha_1^\star} \mathbf{A}_1^\star$ and $\sqrt{\alpha_2^\star} \mathbf{A}_2^\star$.
\end{remark}

In order to extend the multi-terminal case to more than 2 encoders, we need to modify the problem formulation accordingly. Assume that we have $R$ terminals, comprised of $R$ parallel source-to-sensor channel matrices $\{\mathbf{H}_l\}_{l=1}^R$ and noise vectors $\{\mathbf{v}_l\}_{l=1}^R$, $R$ CS encoders $\{\mathbf{A}_l\}_{l=1}^R$, $R$ channels $\{\mathbf{G}_l\}_{l=1}^R$ and $R$ channel noise vectors $\{\mathbf{w}_l\}_{l=1}^R$. Then, equations \eqref{eq:access} and \eqref{eq:cov nosie} are modified by adding the matrices and vectors associated with $R$ terminals. Furthermore, the power constraint in \eqref{eq:power_ortho} would be modified by extending the summation from $l=1$ to $l=R$. Consequently, the optimization in \eqref{eq:opt 1_dist} can be solved with respect to variables $\{\mathbf{A}_l\}_{l=1}^R$. We also note that the equivalent optimization problem in \eqref{eq:opt 1_final_dist} should be modified by introducing $R$ optimization variables $\{\mathbf{Q}_l\}_{l=1}^R$ and $R$ variables $\{\mathbf{Y}_l\}_{l=1}^R$. Similarly, the constraints in \eqref{eq:opt 1_final_dist} should be modified by including the parameters and variables associated with the $R$ terminals.

\subsection{Coherent MAC} \label{sec:problem_coh}
We consider the multi-terminal setup with coherent MAC that is shown in \figref{fig:diagram_dist_coh}. 
\begin{figure} [!ht]
  \centering
    \psfrag{y}[][][0.75]{$\mathbf{y}$}
  \psfrag{x}[][][0.75]{$\mathbf{x}$}
  \psfrag{O}[][][0.75]{$\mathbf{H}_1$}
   \psfrag{P}[][][0.75]{$\mathbf{H}_2$}
  \psfrag{A}[][][0.75]{$\mathbf{v}_1$}
    \psfrag{B}[][][0.75]{$\mathbf{v}_2$}
  \psfrag{L}[][][0.75]{$\mathbf{z}_1$}
  \psfrag{M}[][][0.75]{$\mathbf{z}_2$}
   \psfrag{Q}[][][0.75]{$\mathbf{A}_1$}
  \psfrag{R}[][][0.75]{$\mathbf{A}_2$}
  \psfrag{F}[][][0.75]{$\mathbf{G}_1$}
    \psfrag{G}[][][0.75]{$\mathbf{G}_2$}
    \psfrag{d}[][][0.75]{$\mathbf{w}$}
    \psfrag{u}[][][0.75]{$\mathbf{y}_1$}
   \psfrag{v}[][][0.75]{$\mathbf{y}_2$}
  \psfrag{h}[][][0.75]{$\widehat{\mathbf{x}}$}
  \psfrag{E}[][][0.75]{CS encoder}
  \psfrag{C}[][][0.75]{Channel}
    \psfrag{Z}[][][0.75]{FC}
   \psfrag{D}[][][0.75]{Decoder}
    \includegraphics[width=9cm]{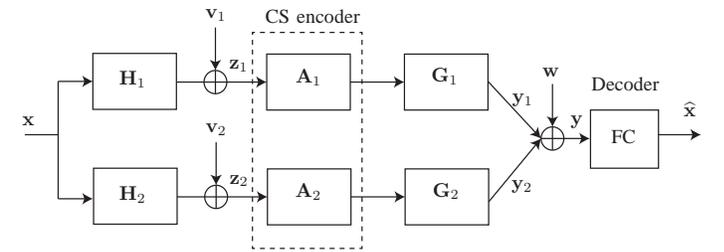}
  \caption{Studied model for multi-terminal system with coherent MAC.}\label{fig:diagram_dist_coh}
  \centering
\end{figure}
The system model using coherent MAC is similar to that of the orthogonal MAC, described in \secref{sec:orthog_mac}, with the difference that the transmitted observations from all terminals are superimposed and received as a coherent sum. We also assume that the size of observations at each terminal are equal, i.e., $M_1 = M_2 \triangleq M$.
The received measurements at the FC can be written as 
\begin{equation} \label{eq:meas_coh_final}
\begin{aligned}
	\mathbf{y} =  \mathbf{y}_1 + \mathbf{y}_2 + \mathbf{w} 
	= \widetilde{\mathbf{A}} \widetilde{\mathbf{H}} \mathbf{x} + \underbrace{\widetilde{\mathbf{A}} \widetilde{\mathbf{v}} + \mathbf{w}}_{\widetilde{\mathbf{n}}}.
\end{aligned}
\end{equation}
where
\begin{equation} \label{eq:access_coh}
\begin{aligned}
	 &\widetilde{\mathbf{A}} \triangleq \left[g_1 \sigma_{v_1}\mathbf{A}_1 \; \; g_2 \sigma_{v_2} \mathbf{A}_2 \right] \; ,  \;
	  \widetilde{\mathbf{H}} &\triangleq \left[\frac{1}{\sigma_{v_1}}\mathbf{H}_1^\top \; \; \frac{1}{\sigma_{v_2}}\mathbf{H}_2^\top \right]^\top,& \\ 
	 &\widetilde{\mathbf{v}} \triangleq \left[\frac{1}{\sigma_{v_1}} \mathbf{v}_1^\top \; \; \frac{1}{\sigma_{v_2}} \mathbf{v}_2^\top \right]^\top .& \\
\end{aligned}
\end{equation}

Denoting the total noise in the system by $\widetilde{\mathbf{n}} \triangleq \widetilde{\mathbf{A}} \widetilde{\mathbf{v}} + \mathbf{w}$, the covariance matrix associated with $\widetilde{\mathbf{n}} $ is
\begin{equation} \label{eq:noise_cov_coh}
	\widetilde{\mathbf{R}}_n \triangleq  \widetilde{\mathbf{A}} \widetilde{\mathbf{A}}^\top + \sigma_w^2 \mathbf{I}_M.
\end{equation}

Following \eqref{eq:oracle MMSE est}, it can be shown that the oracle estimator of $\mathbf{x}$ given the measurements \eqref{eq:meas_coh_final}, i.e., $\mathbb{E}[\mathbf{x} | \mathbf{y}, \mathcal{S}]$, gives the following MSE
\begin{equation} \label{eq:oracle_MSE_dist_coh}
\begin{aligned}
	\mathrm{MSE}_c^{(lb)}
	\! \!  \triangleq  \! \sum_{\mathcal{S}} \! \frac{1}{{N \choose K}} \! \Tr \! \left\{ \! \left(\mathbf{R}^{-1} \!+\! \mathbf{E}_\mathcal{S}^\top  \widetilde{\mathbf{H}}^\top \widetilde{\mathbf{A}}^\top \widetilde{\mathbf{R}}_n^{\! -1} \widetilde{\mathbf{A}} \widetilde{\mathbf{H}}  \mathbf{E}_\mathcal{S} \right)^{\!-1} \!\right\}.
\end{aligned}
\end{equation}

We obtain the average power constraint in the case of coherent MAC as
\begin{equation} \label{eq:power_coh}
\begin{aligned}
	&\sum_{l=1}^2 \mathbb{E}[\| \mathbf{A}_l \mathbf{H}_l  \mathbf{x} + \mathbf{A}_l \mathbf{v}_l \|_2^2] & \\
	&= \sum_{l=1}^2 \Tr \{\mathbf{A}_l\mathbf{H}_l \mathbf{R}_x \mathbf{H}_l^\top \mathbf{A}_l^\top + \sigma_{v_l}^2 \mathbf{A}_l\mathbf{A}_l^\top\} \leq P,&
\end{aligned}
\end{equation}           
where $P>0$ is available power. Further, we used the fact that the source and source-to-sensor noises are uncorrelated as well as the fact that $\mathbb{E}[\widetilde{\mathbf{v}} \widetilde{\mathbf{v}}^\top] = \mathbf{I}_{L_1 + L_2}$. Therefore, we pose the following optimization problem for sensing matrix design
\begin{equation} \label{eq:oracle_MSE_dist_coh_final}
\begin{aligned}
	&\underset{\mathbf{A}_1, \mathbf{A}_2}{\text{minimize}} \hspace{0.25cm} \mathrm{MSE}_c^{(lb)}& \\
	& \text{subject to} \hspace{0.25cm} \sum_{l=1}^2 \Tr \{\mathbf{A}_l\mathbf{H}_l \mathbf{R}_x \mathbf{H}_l^\top \mathbf{A}_l^\top + \sigma_{v_l}^2 \mathbf{A}_l\mathbf{A}_l^\top\} \leq P,& 
\end{aligned}
\end{equation}
where $\mathrm{MSE}_c^{(lb)}$ is shown in \eqref{eq:oracle_MSE_dist_coh}. The following theorem gives an equivalent optimization problem to \eqref{eq:oracle_MSE_dist_coh_final}.

\begin{theorem} \label{theo:mult_ter_coh}
	Let $\widetilde{\mathbf{Q}} = \widetilde{\mathbf{A}}^\top \widetilde{\mathbf{A}}$, then the optimization problem \eqref{eq:oracle_MSE_dist_coh_final} is equivalent to solving

	{\small \begin{equation} \label{eq:opt 1_final_coh}
\begin{aligned}
	&\underset{\widetilde{\mathbf{Q}},\widetilde{\mathbf{X}}_\mathcal{S},\widetilde{\mathbf{Y}}}{\text{minimize}} \hspace{0.25cm} \sum_\mathcal{S} \Tr \{\widetilde{\mathbf{X}}_\mathcal{S}\}	& \\
	&\text{subject to} \hspace{0.25cm} \left[
\begin{array}{c c}
   \mathbf{R}^{-1} + \frac{1}{\sigma_w^2} \widetilde{\mathbf{D}}_\mathcal{S}^\top  \widetilde{\mathbf{Q}} \widetilde{\mathbf{D}}_\mathcal{S}  - \widetilde{\mathbf{D}}_\mathcal{S}^\top \widetilde{\mathbf{Y}} \widetilde{\mathbf{D}}_\mathcal{S} & \mathbf{I}_K \\ 
  \mathbf{I}_K  &   \widetilde{\mathbf{X}}_\mathcal{S}   \\
\end{array}
\right] \succeq \mathbf{0} &\\
	&\hspace{1.5cm}  \left[
	\begin{array}{c c}
	   \widetilde{\mathbf{Y}} & \frac{1}{\sigma_w}\widetilde{\mathbf{Q}} \\ 
	  \frac{1}{\sigma_w}  \widetilde{\mathbf{Q}}  &  \sigma_w^2 \mathbf{I}_{L_1 + L_2}  + \widetilde{\mathbf{Q}} \\
	\end{array}
	\right] \succeq \mathbf{0} , \; \forall \mathcal{S} & \\
	 &\hspace{1.5cm} \sum_{l=1}^2 \Tr \left \{\left(\mathbf{H}_l \mathbf{R}_x \mathbf{H}_l^\top + \sigma_{v_l}^2 \mathbf{I}_{L_l} \right) \mathbf{Q}_l \right\} \leq P & \\
	&\hspace{1.7cm} \widetilde{\mathbf{Q}} \succeq \mathbf{0},  \hspace{0.2cm}  \mathrm{rank}(\widetilde{\mathbf{Q}}) = M,&
\end{aligned}
\end{equation}}
where we have defined $\widetilde{\mathbf{D}}_\mathcal{S} \triangleq \widetilde{\mathbf{H}} \mathbf{E}_\mathcal{S}$, and where $\widetilde{\mathbf{Q}}  \in \mathbb{R}^{(L_1 + L_2) \times (L_1 + L_2)}$, $\widetilde{\mathbf{X}}_\mathcal{S} \in \mathbb{R}^{K \times K}$ and $\widetilde{\mathbf{Y}} \in \mathbb{R}^{(L_1 + L_2) \times (L_1 + L_2)}$ are optimization variables. Further, $\mathbf{Q}_l \in \mathbb{R}^{L_l \times L_l}$ ($l \in \{1,2\}$), is the $l^{th}$ diagonal block of $\widetilde{\mathbf{Q}}$.
\end{theorem}

\begin{proof}
	The proof is omitted since it can be followed by the proofs of \theoref{theo:sing_ter} and \theoref{theo:mult_ter_ortho}.
\end{proof}

\begin{remark} \label{rem:coh}
In order to solve the optimization problem \eqref{eq:opt 1_final_coh} for  $\mathbf{A}_l$, $l \in \{1,2\}$, we follow similar steps as in \procref{proc}: We first relax the problem \eqref{eq:opt 1_final_coh} by ignoring the rank constraint, which results in a convex SDR program with respect to optimization variables. Once $\widetilde{\mathbf{Q}}^\star$ is determined, we take the EVD $\widetilde{\mathbf{Q}}^\star = \mathbf{U}_{\widetilde{q}} \mathbf{\Gamma}_{\widetilde{q}} \mathbf{U}_{\widetilde{q}}$, then approximately reconstruct $\widetilde{\mathbf{A}}^\star = \mathbf{U}_{\widetilde{a}} [\mathrm{diag}(\sqrt{\gamma_{\widetilde{q}_1}},\ldots,\sqrt{\gamma_{\widetilde{q}_M}}) \; \; \mathbf{0}_{M \times (L_1 + L_2 - M)}] \mathbf{U}_{\widetilde{q}}^\top$, where $\gamma_{\widetilde{q}_i}$ ($1 \leq i \leq M$) are the largest eigen-values of $\mathbf{Q}^\star$. Next, we partition $\widetilde{\mathbf{A}}^\star$ to extract matrices $\mathbf{A}_l^\star$, $l \in \{1,2\}$. For power-rescaling the sensing matrices $\mathbf{A}_l^\star$ to meet the power constraint, similar to the orthogonal MAC, we give the weighting coefficient $\sqrt{\alpha_l}$ to the corresponding matrix and optimize over ${\alpha_l}$. 
The optimization is done by solving \eqref{eq:opt 1_final_coh} with new optimization variables $\alpha_1 \geq 0$ and $\alpha_2 \geq 0$. Note that the rank and positive semi-definite constraints are immaterial at this stage since $\mathbf{\widetilde{Q}}$  already fulfils these constraints. In this case, $\widetilde{\mathbf{Q}}$ becomes
\begin{equation*}
\begin{aligned}
	\widetilde{\mathbf{Q}}  = \widetilde{\mathbf{A}}^{\star \top} \widetilde{\mathbf{A}}^\star 
	 = \left[
	\begin{array}{c c}
	   \alpha_1 \sigma_{v_1}^2  \mathbf{A}_1^{\star \top} \mathbf{A}_1^\star & \alpha_3   \sigma_{v_1} \sigma_{v_2} \mathbf{A}_1^{\star \top} \mathbf{A}_2^\star \\ 
	  \alpha_3  \sigma_{v_1} \sigma_{v_2} \mathbf{A}_2^{\star \top} \mathbf{A}_1^\star  &  \alpha_2  \sigma_{v_2}^2  \mathbf{A}_2^{\star \top} \mathbf{A}_2^\star\\
	\end{array}
	\right],
\end{aligned}
\end{equation*}
where $\alpha_3 = \sqrt{\alpha_1 \alpha_2}$, and $\mathbf{A}_l^\star$ is known from the previous stage. In order to convexify the latter assumption, using the Schur's complement \cite{04:Boyd_book}, we write it as the following matrix inequality
\begin{equation*}
\begin{aligned}
	 \left[
	\begin{array}{c c}
	   \alpha_1 & \alpha_3 \\ 
	  \alpha_3   &  \alpha_2 \\
	\end{array}
	\right] \succeq \mathbf{0}.
\end{aligned}
\end{equation*}
Hence, the power-rescaling optimization problem becomes convex in variables $\alpha_1 \geq 0$, $\alpha_2 \geq 0$, $\alpha_3 \geq 0$, $\widetilde{\mathbf{X}}_\mathcal{S}$ and $\widetilde{\mathbf{Y}}$ which can be solved using any standard convex solver. Note also that the final rescaled optimized sensing matrices would be $\sqrt{\alpha_1^\star} \mathbf{A}_1^\star$ and $\sqrt{\alpha_2^\star} \mathbf{A}_2^\star$ which satisfy the power constraint with equality. 
\end{remark}

The extension of the design procedure for coherent MAC with more than 2 terminals is straightforward, and can be done using the same steps as discussed in the previous subsection for orthogonal MAC.

\section{Complexity Considerations} \label{sec:complexity}
In this section, we discuss the computational complexity of solving the proposed optimization scheme for sensing matrix design in single-- and multi--terminal settings. We also provide a low-complexity alternative design approach based on stochastic optimization.

First, in the single-terminal setting, we note that the high computational complexity in \procref{proc} arises from the first step, i.e., solving the SDR problem (\eqref{eq:opt 1_final}  without the rank constraint). More precisely, the SDR problem consists of one matrix variable $\mathbf{Q}$ of size $L \times L$,  ${N \choose K}$  matrix variables $\mathbf{X}_\mathcal{S}$ of size $K \times K$, and one matrix variable $\mathbf{Y}$ of size $L \times L$. Hence, it can be iteratively solved using interior point methods with computational complexity growing at most like $\mathcal{O}(2 L^6 + {N \choose K}^3 K^6)$ arithmetic operations in each iteration \cite{06:Zhi-Quan}. Following similar arguments, the computational complexity of solving the SDR problems associated with multi-terminal orthogonal MAC, i.e., \eqref{eq:opt 1_final_dist}, and multi-terminal coherent MAC, i.e., \eqref{eq:opt 1_final_coh}, grows at most like $\mathcal{O}(2 L_1^6 + 2 L_2^6 + {N \choose K}^3 K^6)$ and $\mathcal{O}(2 (L_1 + L_2)^6 + {N \choose K}^3 K^6)$, respectively. Therefore, it can be seen that as $N$ increases, the computational complexity grows exponentially\footnote[1]{Note that ${N \choose K} \approx 2^{N H(K/N)}$, where $H(\cdot)$ denotes the binary entropy function, i.e., $H(p) \triangleq -p \log_2 p - (1 - p) \log_2(1-p)$, for $0 < p < 1$.}. 

The computational complexity of solving the SDR problems associated with \eqref{eq:opt 1_final}, \eqref{eq:opt 1_final_dist} and \eqref{eq:opt 1_final_coh} can be significantly reduced under certain assumptions (see, e.g., the special cases I-IV in \secref{sec:special}), for which closed-form solutions can be derived. Here, we offer an alternative in order to solve the SDR problem of \eqref{eq:opt 1_final} in a less computational way. Note that the objective function $\mathrm{MSE}^{(lb)}$ in \eqref{eq:oracle_MSE_2} can be rewritten as 
\begin{equation} \label{eq:less_complex}
	\mathrm{MSE}^{(lb)} = \mathbb{E}_{\boldsymbol{\mathcal{S}}} \left[  \Tr \left\{ \left(\mathbf{R}^{-1} \!+\! g^2 \mathbf{E}_{\boldsymbol{\mathcal{S}}}^\top \mathbf{H}^\top \mathbf{A}^\top \mathbf{R}_n^{\! -1} \mathbf{A} \mathbf{H} \mathbf{E}_{\boldsymbol{\mathcal{S}}} \right)^{\!-1} \right\} \right],
\end{equation}
where $\boldsymbol{\mathcal{S}}$ is a random variable which picks a support set $\mathcal{S}$ uniformly at random from the set of all possibilities $\Omega$, and $\mathbb{E}_{\boldsymbol{\mathcal{S}}}$ denotes the expectation  with respect to  the random support set  $\boldsymbol{\mathcal{S}}$. Notice that the expectation in \eqref{eq:less_complex} can be (approximately) calculated using the sample mean as
\begin{equation} \label{eq:less_complex_2}
	\mathrm{MSE}^{(lb)} \! \approx \!\frac{1}{|\Omega'|} \sum_{\mathcal{S}' \subset \Omega'} \!   \Tr \left\{ \left(\mathbf{R}^{-1} \! \!+\! g^2 \mathbf{E}_{\mathcal{S}'}^\top \mathbf{H}^\top \mathbf{A}^\top \mathbf{R}_n^{\! -1} \mathbf{A} \mathbf{H} \mathbf{E}_{\mathcal{S}'} \right)^{\!-1} \right\} 
\end{equation}
where $\mathcal{S}'$ is uniformly chosen from a set $\Omega' \subset \Omega$. Note that the cardinality $|\Omega'|$ can be chosen to be far less than ${N \choose K}$. As a result, the computational complexity of solving the resulting SDR problem reduces to $\mathcal{O}(2 N^6 + |\Omega'|^3 K^6)$ arithmetic operations, where $ |\Omega'| \ll {N \choose K}$. Following the same arguments, the SDR problems of \eqref{eq:opt 1_final_dist} and \eqref{eq:opt 1_final_coh} can be also approximately solved with a significantly reduced computational complexity.

\begin{remark}
We note that in the above analysis, we assume that all the support sets are uniformly drawn from ${N \choose K}$ possibilities, i.e., all supports are equi-probable. Hence, according to \eqref{eq:less_complex}, there is no preference towards selecting a particular sparsity pattern in order to use the sample-mean approximation in \eqref{eq:less_complex_2}. However, by choosing a larger number of sparsity patterns, the approximation becomes tighter due to the {\em Law of Large Numbers}. It should be mentioned that the uniformly random selection of the support sets is indeed the worst-case assumption. If the support sets are selected according to a different non-uniform distribution, say $q(\mathcal{S})$, then one can approximate \eqref{eq:less_complex} by neglecting the tail of the probability density function $q(\mathcal{S})$. Owing to the concentration inequalities, the probability that the selected pattern $\boldsymbol{\mathcal{S}}$ exceeds the sum of mean and two/three times the standard deviation of this distribution is small, and the support set patterns for averaging can be chosen accordingly. 
\end{remark}

\section{Numerical Experiments} \label{sec:sim}
For the single--terminal setting, we provide numerical experiments for evaluating the sensing matrix design scheme proposed in \procref{proc}, which is referred to as
\begin{itemize}
	\item \textit{Lower-bound minimizing sensing matrix (\procref{proc})},
\end{itemize}
and compare it with the following design methods:
\begin{itemize}
	\item \textit{Upper-bound minimizing sensing matrix:} Using this method, we upper-bound the MSE of the MMSE estimator of the sparse source vector by that of the linear MMSE (LMMSE) estimator. The MSE incurred by using the LMMSE estimator can be written as
\begin{equation*} \label{eq:ub_MSE}
\begin{aligned}
	\mathrm{MSE}^{(ub)} &\triangleq  \Tr \left\{ \left(\mathbf{R}_x^{-1} + g^2 \mathbf{H}^\top \mathbf{A}^\top \mathbf{R}_n^{-1} \mathbf{A} \mathbf{H} \right)^{-1} \right\}.&
\end{aligned}
\end{equation*}
Optimizing the sensing matrix with respect to minimizing the above equation under a power constraint has been studied in \cite{08:Jin,07:Schizas}.
	
	\item \textit{Gaussian sensing matrix:} This method is typically a standard approach in literature for generating a sensing matrix. Each element of the Gaussian sensing matrix is generated according to the standard Gaussian distribution. 
		\item \textit{Tight frame}: Using this method, the sensing matrix is chosen as $\mathbf{A} =  \mathbf{U}_a \left[\mathbf{I}_M \; \; \mathbf{0}_{M \times (L-M)}\right]  \mathbf{V}_a^\top$, where $\mathbf{U_a} \in \mathbb{R}^{M \times M}$ and $\mathbf{V}_a \in \mathbb{R}^{L \times L}$ are arbitrary unitary matrices. 
\end{itemize}

Note that we scale the resulting sensing matrix, described above, by $\sqrt{P / \Tr\{(\mathbf{H} \mathbf{R}_x \mathbf{H}^\top+ \sigma_v^2 \mathbf{I}_L)\mathbf{A}^{\top} \mathbf{A}\}}$ in order to satisfy the power constraint. We also compare the actual MSE, incurred by using the above methods, with the value of the lower-bound \eqref{eq:oracle_MSE_2} when the lower-bound minimizing sensing matrix is applied. This will be referred to as \textit{lower-bound} in our experiments. It should be also mentioned that for solving the convex SDR problems, we use the \texttt{CVX} solver \cite{cvx} . 

We also compare the performance of the proposed schemes for the single--terminal setting, and multi--terminal settings with orthogonal and coherent MAC described in \remref{rem:ortho} and \remref{rem:coh}, respectively. 

\subsection{Experimental Setups} \label{sec:setup}

We evaluate the performance using the normalized MSE (NMSE) criterion, defined as\footnote[2]{NMSE can be thought of as MSE per degree of freedom.}
\begin{equation*}
	\mathrm{NMSE} \triangleq \frac{\mathbb{E}[\|\mathbf{x} - \widehat{\mathbf{x}}\|_2^2]}{K},
\end{equation*}
where $\widehat{\mathbf{x}}$ is the decoder's output. 

In addition to NMSE, we also compare the performance of proposed sensing matrix design in terms of the \textit{probability of support set recovery} which is defined as
\begin{equation*}
	\mathrm{Pr}\{n \neq \hat{n}: n \in \mathcal{S}, \hat{n} \in \hat{\mathcal{S}}\},
\end{equation*}
where $\hat{\mathcal{S}}$ is the reconstructed support set of the vector $\widehat{\mathbf{x}}$, and $\hat{n} \in \{1,2,\ldots,N\}$ is an element of the reconstructed support set $\hat{\mathcal{S}}$.

Our simulation setup is described as follows.
  For given values of sparsity level $K$ (assumed known in advance) and input vector size $N$, we choose the number of measurements $M$. 
   We randomly generate a set of exactly $K$-sparse vector $\mathbf{x}$, where the support set $\mathcal{S}$ with $|\mathcal{S}| = K$ is chosen uniformly at random over the set $\{1,2,\ldots,N\}$. The non-zero components of $\mathbf{x}$ are drawn from Gaussian distribution $\mathcal{N}(\mathbf{0},\mathbf{R})$, and the covariance matrix $\mathbf{R} \in \mathbb{R}^{K \times K}$ is generated according to the exponential model \cite{01:Loyka}, where each entry at row $i$ and column $j$ is chosen as $\rho^{| i - j |}$ in which $0 \leq \rho < 1$ is known as correlation coefficient. We compute sample covariance matrix for the sparse source vector, i.e., $\mathbf{R}_x = \mathbb{E}[\mathbf{x} \mathbf{x}^\top]$ using $10^5$ randomly generated samples of the source vector $\mathbf{x}$. We let $L=N$, $\mathbf{H} = \mathbf{I}_N$ and $\mathbf{v= 0}$ for the single--terminal setting, and for each terminal in the multi--terminal setting, 
  and estimate the source $\mathbf{x}$ from noisy measurements using sparse reconstruction algorithms. We mainly use the greedy orthogonal matching pursuit (OMP) algorithm \cite{07:Tropp}, and the Bayesian-based random--OMP reconstruction algorithm \cite{09:Elad}, which is a low-complexity approximation of the exact (exhaustive) MMSE estimator.

\subsection{Experimental Results}

To assess the actual performance of the proposed design methods using Monte-Carlo simulations, we generate $5000$ realizations of the input sparse vector $\mathbf{x}$. In our first two experiments, we use, at the decoder, the random-OMP algorithm for reconstruction of sparse source vector. 

In our first experiment, we use the simulation parameters $N = 36, K= 3, P=10 \text{ dB}, g = 0.5, \sigma_w = 0.1, \rho = 0.25$.
We plot the NMSE of the design methods as a function of $M$ in \figref{fig:MSE_M_randOMP_new}. The value of $M$ can be thought of as bandwidth or number of transmissions over channel. We observe that at all measurement regions, the proposed lower-bound minimizing sensing matrix outperforms the other competing methods by taking into account sparsity pattern of the sparse source. As expected, as the number of measurements increases, the performance of the methods improves, however, it finally saturates and increasing $M$ further does not help to improve NMSE. This is because at higher number of measurements, the NMSE is influenced more by the additive noise which is fixed. As $M$ increases, the performance of the tight frame approaches that of the lower-bound  minimizing sensing matrix, illustrating that the latter behaves like an orthogonal transform. 
\begin{figure}
  \begin{center}
  \includegraphics[width=0.9\columnwidth,height=7cm]{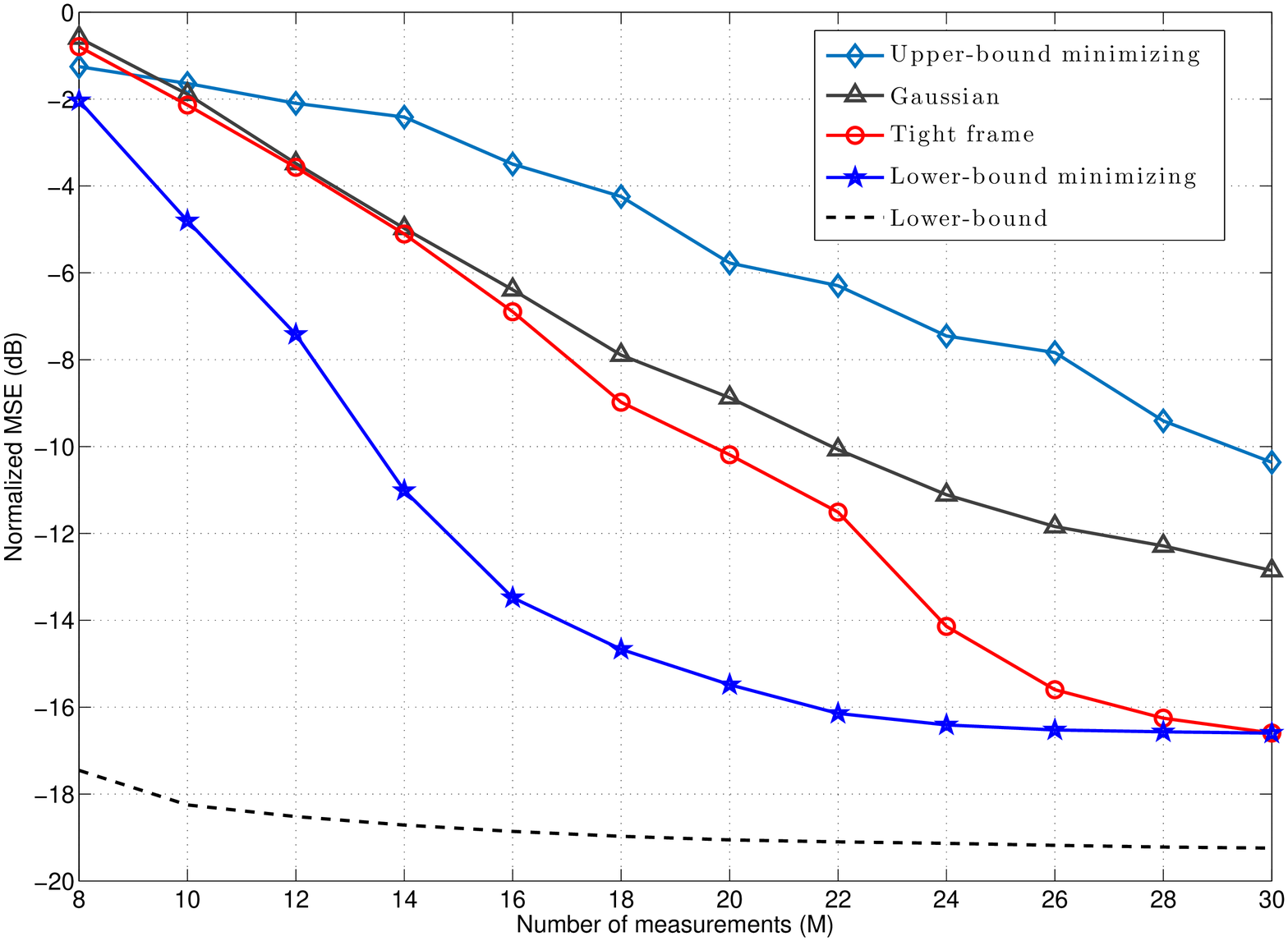}\\
  \caption{NMSE (in dB) as a function of number of measurements $M$ using different sensing matrix design schemes. }
  \label{fig:MSE_M_randOMP_new}
  \end{center}
\end{figure}

Using the same simulation parameters, by fixing $M = 18$, we now vary transmission power $P$ (in dB), and evaluate the performance of the methods in terms of NMSE. The results are reported in \figref{fig:MSE_P_randOMP_new}. In the  low power regime, the performance of the competing methods are almost the same, however, as $P$ increases, the proposed lower-bound minimizing sensing matrix outperforms the other schemes. For example, at $P = 10$ dB, the proposed scheme gives a better performance by more than 6 dB as compared to the other methods.  

\begin{figure}
  \begin{center}
 \includegraphics[width=0.9\columnwidth,height=7cm]{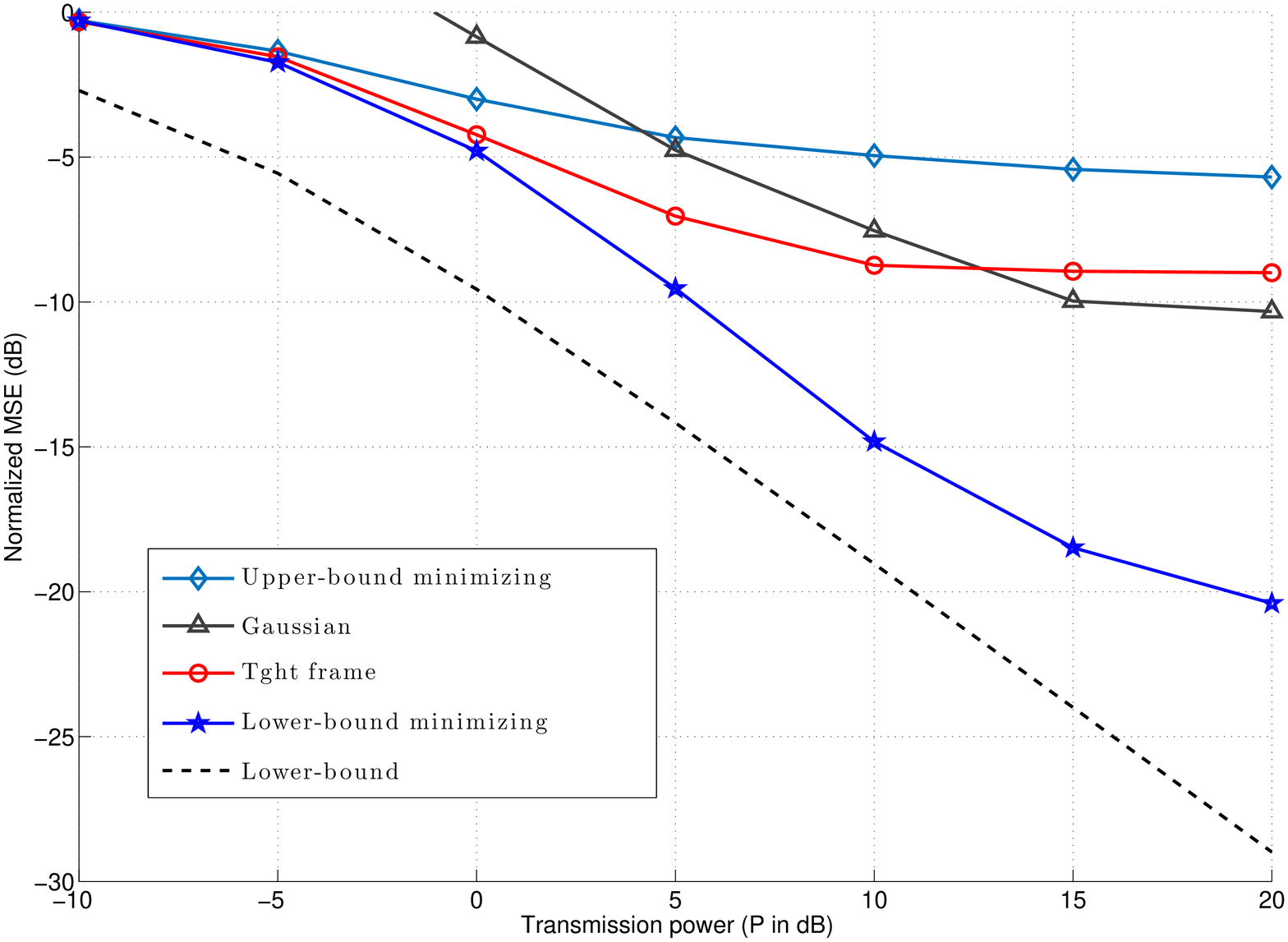}\\
  \caption{NMSE (in dB) as a function of transmission power $P$ (in dB) using different sensing matrix design schemes. }
  \label{fig:MSE_P_randOMP_new}
  \end{center}
\end{figure}

In the previous experiments, we have used the random-OMP algorithm (as the approximation of the exact MMSE estimator) for reconstructing the sparse source. While this algorithm is nearly optimal (in MSE sense), the reconstructed vector might not be necessarily a sparse vector \cite{09:Elad}. In some applications, together with reconstruction accuracy, one might desire a sparse representation at the receiving-end. This, for example, is relevant for compression or recognition purposes. Therefore, in our next experiments, we use the greedy OMP algorithm \cite{07:Tropp} which preserves the sparse structure through reconstructing the source at the decoder's output. 

Setting the decoder as the OMP algorithm, we compare the performance of the methods (in terms of NMSE) as a function of channel signal to noise ratio (CSNR), defined as $\mathrm{CSNR} \triangleq g^2 / \sigma_w^2$, in logarithmic scale. The results are reported in \figref{fig:MSE_G_OMP_new}. Simulation parameters are chosen as $N = 36, K = 3, P = 10 \text{ dB}, M = 18, \rho = 0.5$. We fix $\sigma_w = 0.1$, and vary the CSNR from $1$ to $10^3$ where the channel gain $g$ is chosen accordingly. It is observed that at $\mathrm{CSNR} = 10^2$, the lower-bound minimizing sensing matrix outperforms the Gaussian sensing matrix by more than 8 dB, and the upper-bound minimizing sensing matrix by more than 10 dB. Further, as channel condition improves, the lower-bound minimizing scheme, as compared to other schemes, takes a better advantage of the channel condition in order to reduce the NMSE.  

\begin{figure}
  \begin{center}
 \includegraphics[width=0.9\columnwidth,height=7cm]{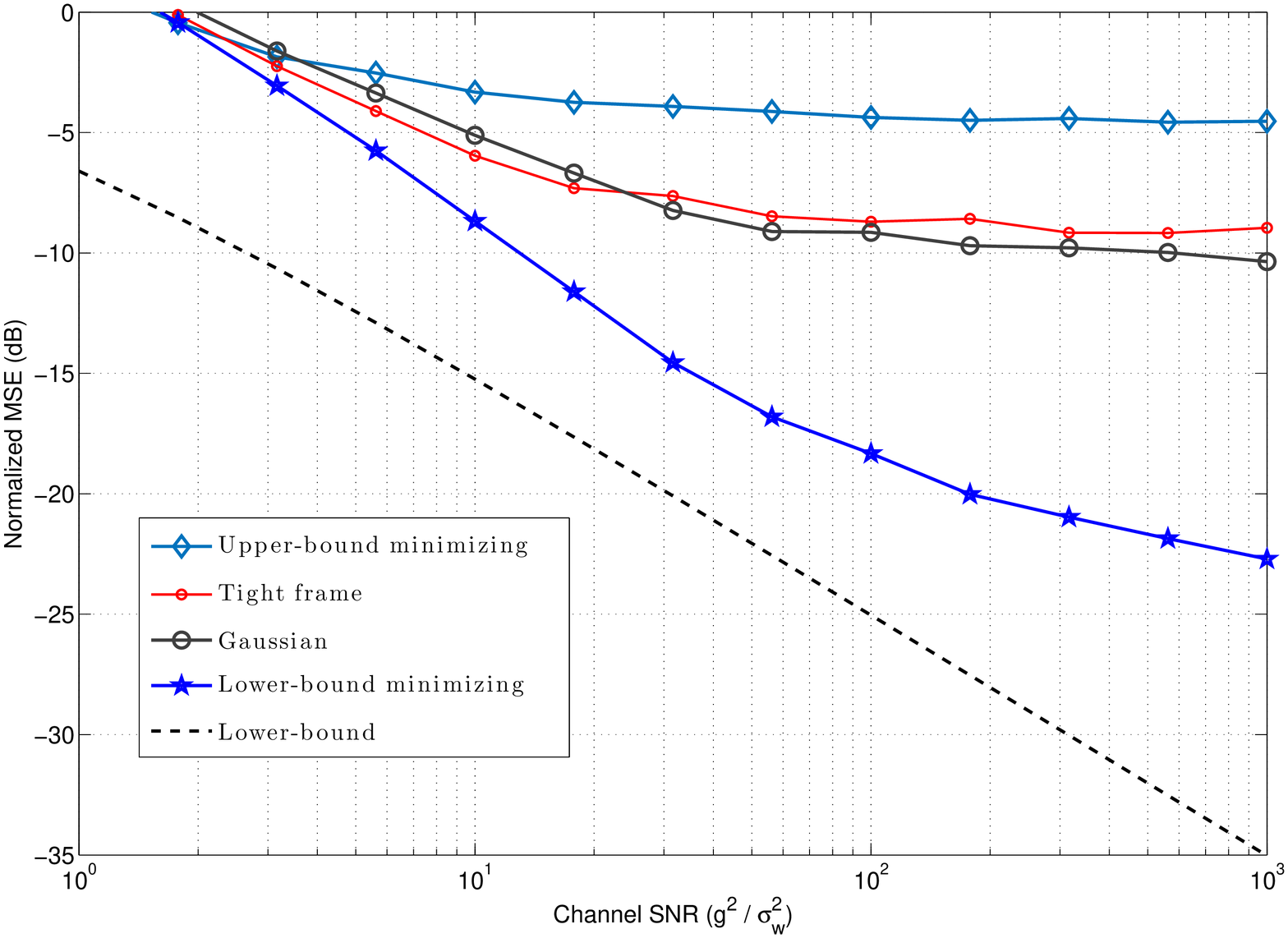}\\
  \caption{NMSE (in dB) as a function of CSNR (in logarithmic scale) using different sensing matrix design schemes. }
  \label{fig:MSE_G_OMP_new}
  \end{center}
\end{figure}

Although the MSE criterion is an important measure of accuracy in performance analysis, the probability of support set recovery is also of central interest in sparse source reconstruction. Therefore, in our next two experiments, we compare the performance of the sensing matrix designs in terms of support set recovery using the OMP algorithm by varying number of measurements (at fixed $P=10$ dB) in \figref{fig:prob_supp_M}, and by varying transmission power (at fixed $M=18$) in \figref{fig:prob_supp_P}. We use the same simulation parameters as those chosen for the previous experiment.

\begin{figure}
  \begin{center}
 \includegraphics[width=0.9\columnwidth,height=7cm]{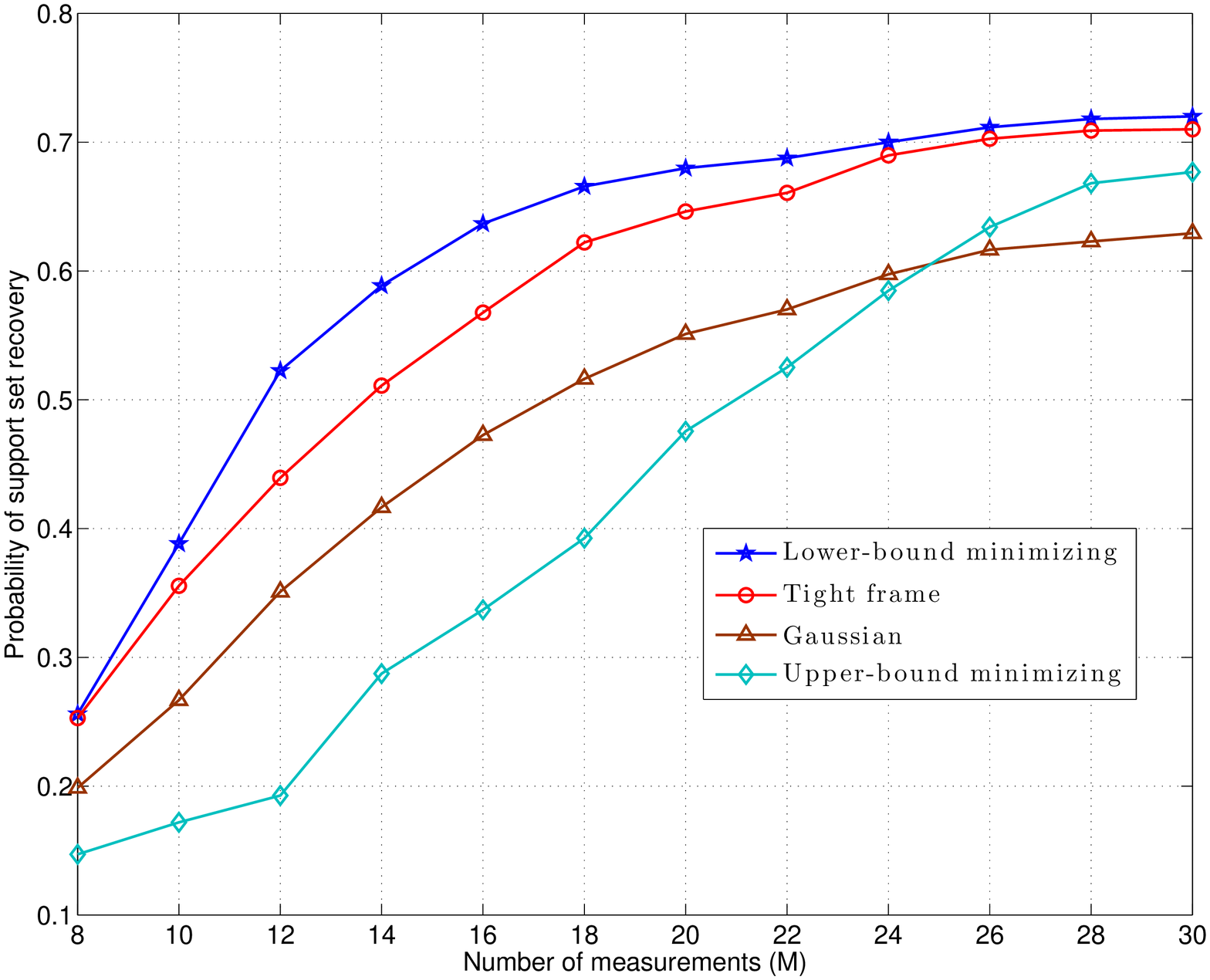}\\
  \caption{Probability of support set recovery as a function of number of measurements using different sensing matrix design schemes. }
  \label{fig:prob_supp_M}
  \end{center}
\end{figure}

\begin{figure}
  \begin{center}
 \includegraphics[width=0.9\columnwidth,height=7cm]{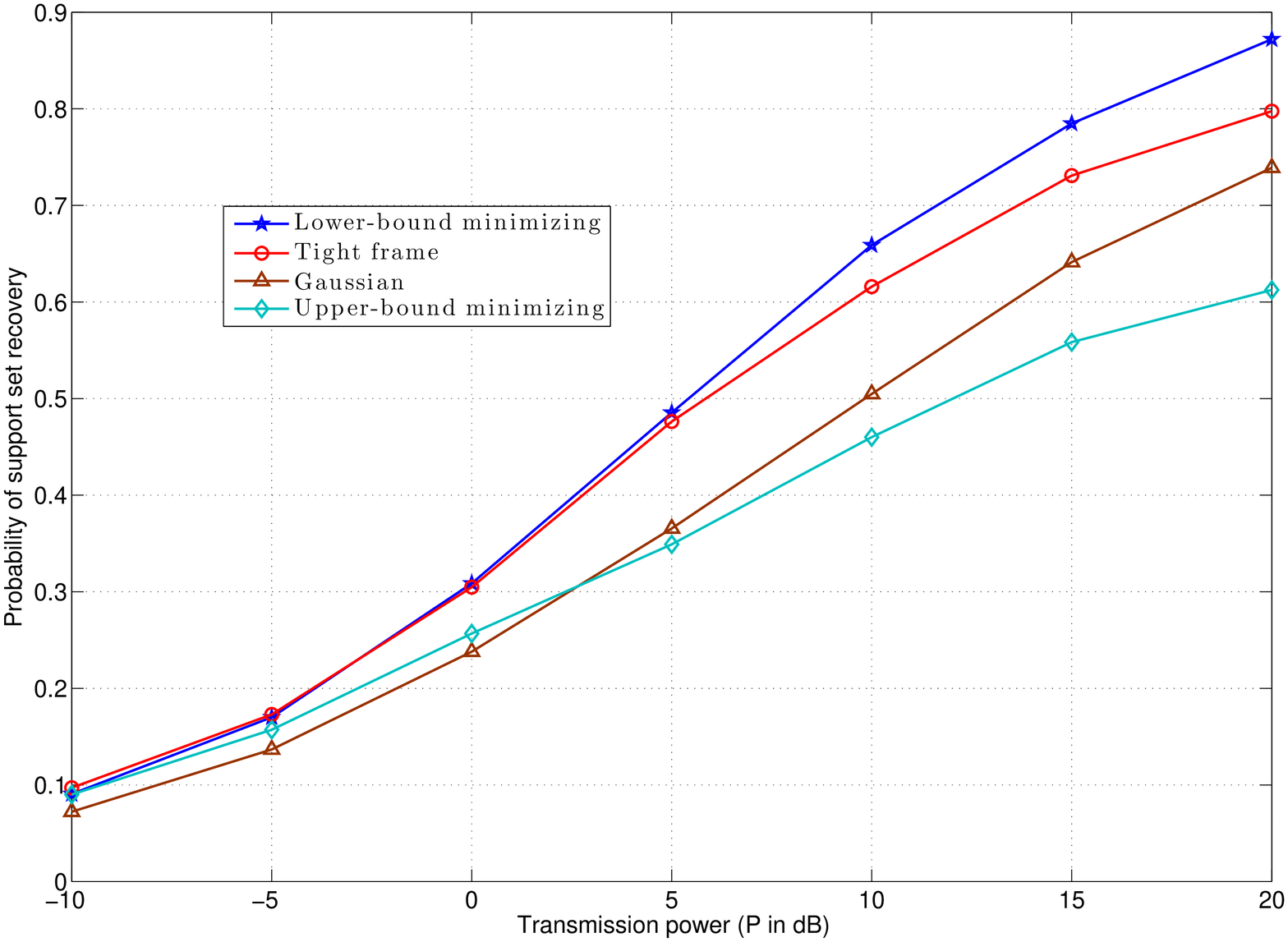}\\
  \caption{Probability of support set recovery as a function of transmission power using different sensing matrix design schemes. }
  \label{fig:prob_supp_P}
  \end{center}
\end{figure}

We observe that the lower-bound minimizing sensing matrix  improves the probability of support set recovery using the OMP reconstruction algorithm. One reason for this behavior is due to the fact that the proposed design endeavors to decrease the mutual coherence $\mu$ of the sensing matrix as discussed in \secref{sec:rel_mu_coh}. The value of $\mu$, which can be calculated by \eqref{eq:mutual co} numerically, at fixed $M = 20$ and $P = 10$ dB is $\mu = 0.46$ for the proposed sensing matrix design, while its value is $0.59$, $0.61$ and $0.75$ for tight-frame, upper-bound minimizing and Gaussian sensing matrices, respectively.

Next, we implement a higher-dimensional system, and apply the proposed low-complexity approach introduced in \secref{sec:complexity}. For this purpose, we choose the following simulation parameters: $N = 100, K = 5, \sigma_w = 0.1, g = 0.5, P =10 \text{ dB}, \rho = 0.75$, and plot the NMSE by varying $M$ in \figref{fig:MSE_M_100_new}. Further, the cardinality of the set $\Omega'$ in \eqref{eq:less_complex_2} is set to 2500, while the cardinality of the set of all sparsity patterns is $| \Omega |  = {N \choose K} \approx 7.5 \times 10^7$. It can be observed while the computational complexity of the lower-bound minimizing scheme has been considerably reduced, it still outperforms the other methods.

\begin{figure}
  \begin{center}
 \includegraphics[width=0.9\columnwidth,height=7cm]{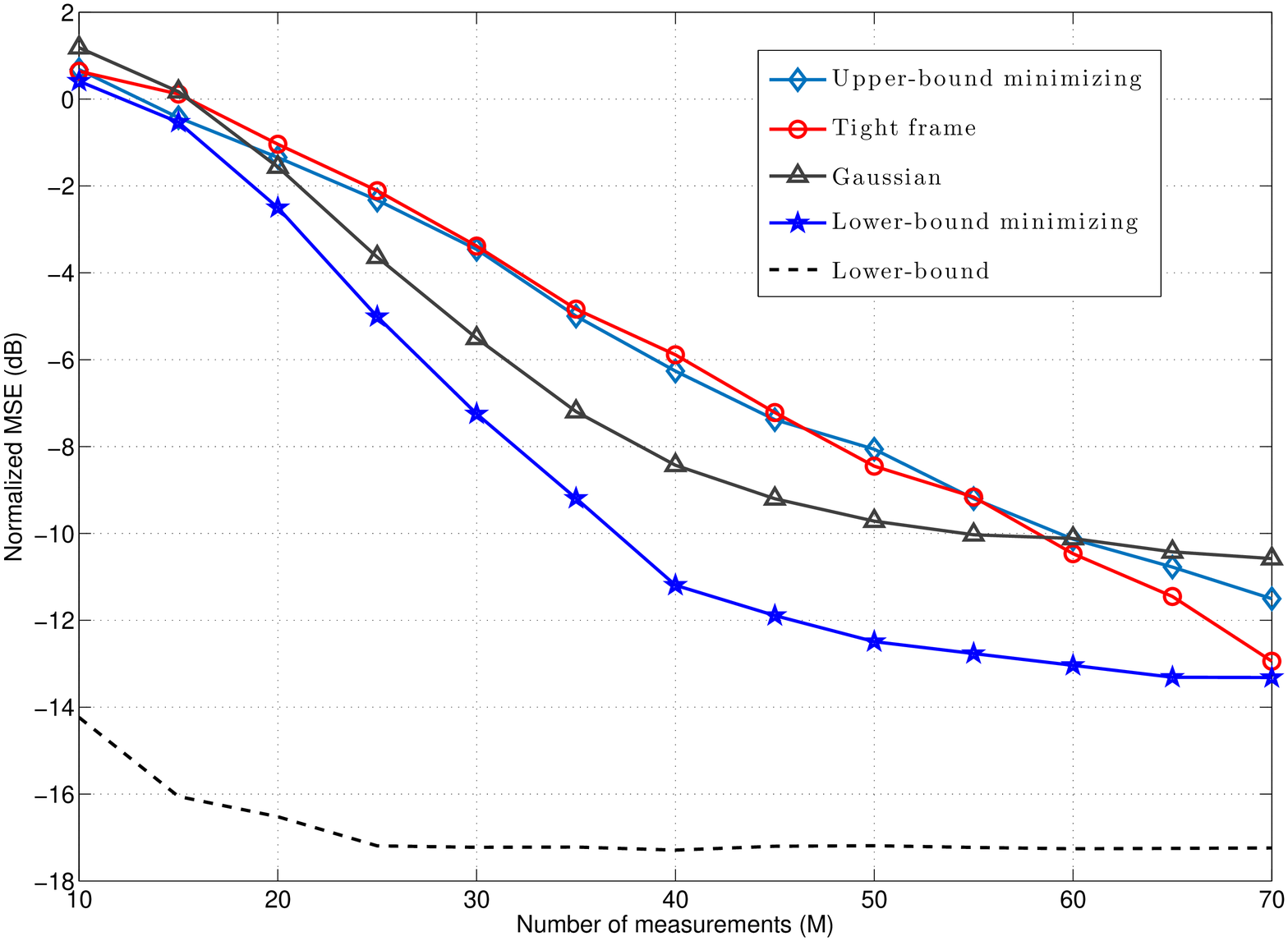}\\
  \caption{NMSE (in dB) as a function of number of measurements using different sensing matrix design schemes. }
  \label{fig:MSE_M_100_new}
  \end{center}
\end{figure}

In our last two experiments, we illustrate the performance of the proposed schemes for multi-terminal settings with orthogonal and coherent MAC. First, we choose simulation parameters as $N = 32, K = 3, \sigma_{w_1} = \sigma_{w_2} = \sigma_w = 0.2, g_1 = 0.5, g_2 = 0.75, P = 10 \text { dB}, \rho = 0.5$, and plot NMSE as a function of number of measurements in \figref{fig:power_scale}, where we assume that $M_1 = M_2$. We compare the performance of the proposed scheme for the orthogonal and coherent MAC with optimized power-rescaling (as described in \remref{rem:ortho} and \remref{rem:coh} by optimizing scaling coefficients $\alpha_1$ and $\alpha_2$), and with unoptimized power-rescaling where $\alpha_1 = \alpha_2$. As can be seen, while optimizing the scaling weights are effective in improving the performance in the coherent MAC, the performance in the orthogonal MAC is not too sensitive to the optimized weights. Further, the performance in the coherent MAC is superior to that of in the orthogonal MAC since, in the latter, each terminal is subject to additive channel noise. 

\begin{figure}
  \begin{center}
 \includegraphics[width=0.9\columnwidth,height=7cm]{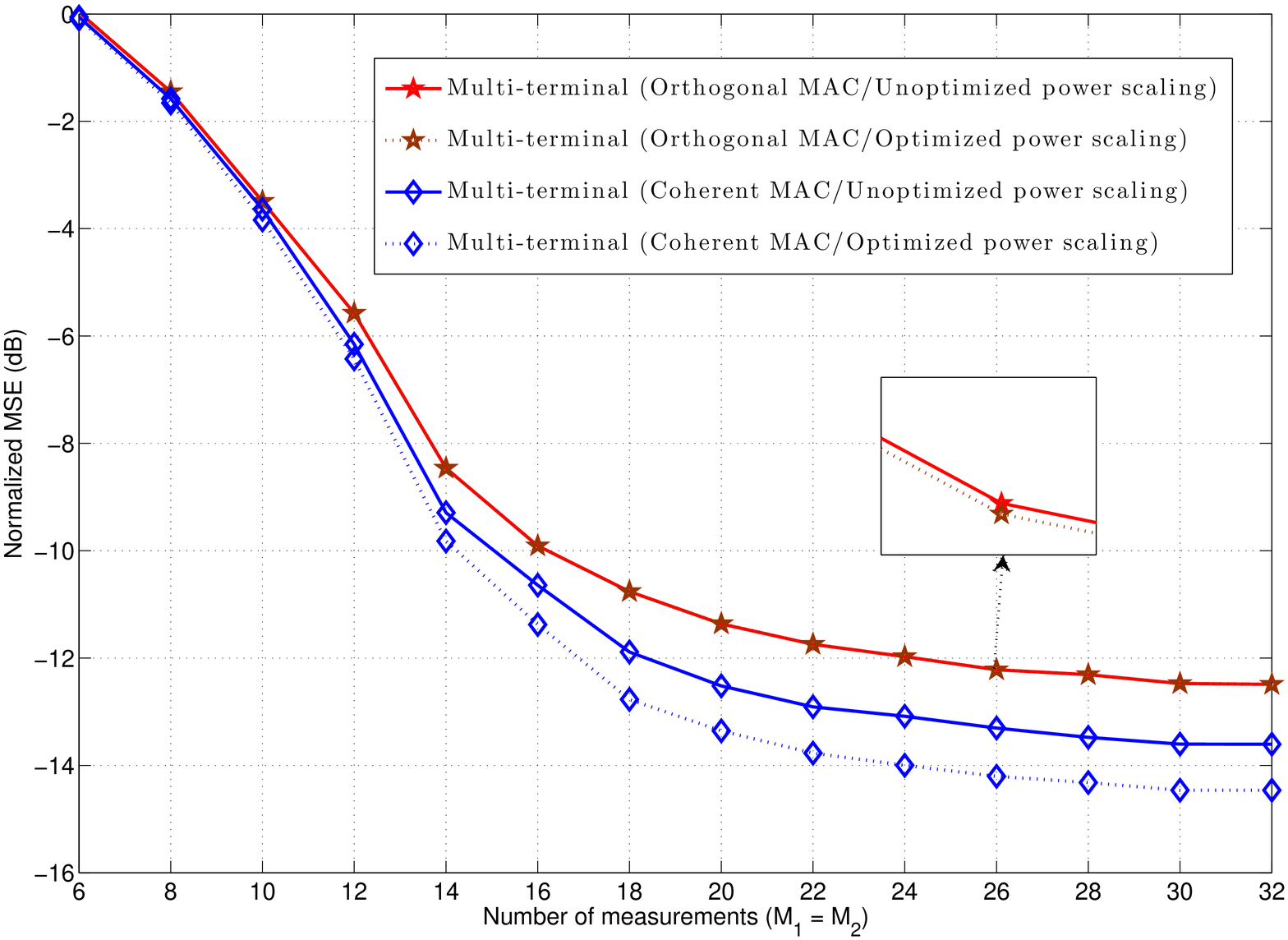}\\
  \caption{NMSE (in dB) as a function of number of measurements in multi--terminal setting with orthogonal and coherent MAC.}
  \label{fig:power_scale}
  \end{center}
\end{figure}

The final experiment demonstrates how a second terminal helps to improve the performance. For this purpose, we compare the proposed low-complexity design methods for the single--terminal setting and  multi--terminal settings with orthogonal MAC and  coherent MAC. In \figref{fig:MSE_GG_multi}, we compare the NMSE (in dB) of the proposed methods as a function of channel gain ratio $g_2 / g_1$ along with their corresponding lower-bounds. We set the following simulation parameters: $N = 64, K = 4, M = 40, \sigma_{w_1} = \sigma_{w_2} = \sigma_w = 0.02, P = 10 \text{ dB}, \rho = 0.5$, and choose $g_1 = 0.5$, then vary the ratio $g_2/g_1$ from $0.5$ to $4$. It can be seen as the channel condition in the second terminal improves, the gap between the performance in single-terminal and multi-terminal settings increases. 
\begin{figure}
  \begin{center}
 \includegraphics[width=0.9\columnwidth,height=7cm]{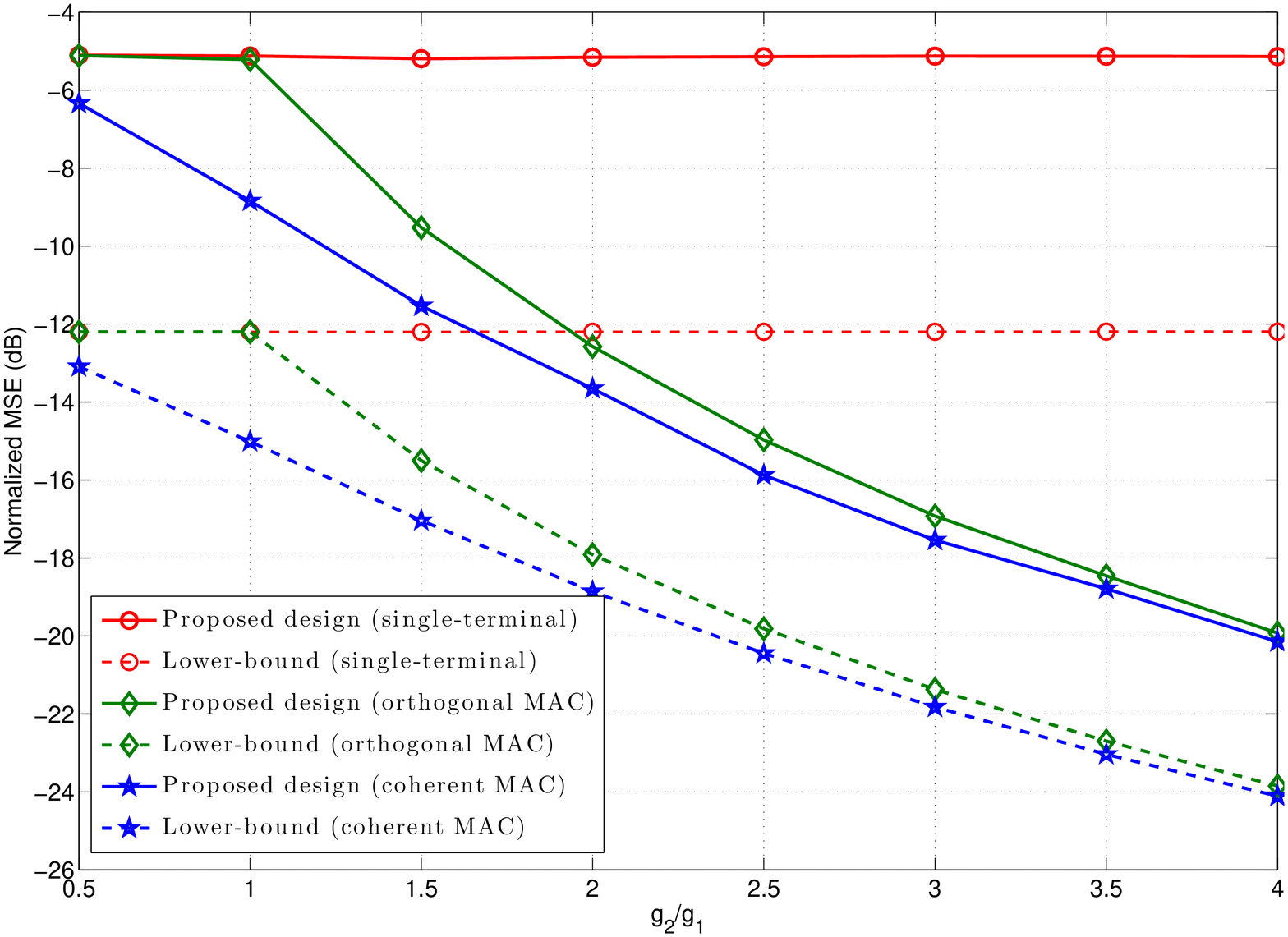}\\
  \caption{NMSE (in dB) as a function of channel gain ratio $g_2/g_1$ in single-- and multi--terminal settings.}
  \label{fig:MSE_GG_multi}
  \end{center}
\end{figure}
\section{Conclusions} \label{sec:conclusions}
We have proposed an optimization procedure for designing sensing matrix, under power constraint, in CS framework and in single-- and multi--terminal (with orthogonal and coherent MAC) settings. The design aims to minimize a lower-bound on MSE of sparse source reconstruction in the studied settings. Under certain conditions, we have been able to address the optimization procedure by deriving closed-form expressions for the sensing matrix. Numerical results show the advantage of our proposed design compared to other relevant schemes in terms of MSE and probability of support set recovery. This advantage has been achieved at the price of higher computational complexity. Therefore, we proposed an approximate optimization procedure in order to reduce the complexity burden. 


\begin{appendices}
\section{Some Useful Lemmas} \label{sec:lemmas}

The following lemmas are stated without proof.
\begin{lemma} \label{lem:prop}
The matrix $\mathbf{E}_\mathcal{S} \!  \in \! \mathbb{R}^{N \! \times  \! K}$, which is formed by taking an identity matrix of order $N \! \times \! N$ and deleting the columns indexed by the support set $\mathcal{S}$, has the following properties:
\begin{itemize}
	\item $\mathbf{E}_\mathcal{S}^\top \mathbf{E}_\mathcal{S} = \mathbf{I}_K$,
	\item $\sum_{\mathcal{S}} \mathbf{E}_\mathcal{S} \mathbf{E}_\mathcal{S}^\top = \frac{{N \choose K}}{K} \mathbf{I}_N.$
\end{itemize}
\end{lemma}

\begin{lemma} \label{lem:deriv_cov}
The covariance matrix of the sparse source, i.e., $\mathbf{R}_{x}$, can be parametrized as
\begin{equation} \label{deriv_cov}
\begin{aligned}
	\mathbf{R}_x = \frac{1}{{N \choose K}}\sum_\mathcal{S} \mathbf{E}_\mathcal{S} \mathbf{R} \mathbf{E}_\mathcal{S}^\top,
\end{aligned}
\end{equation}
where $\mathbf{R}$ is the covariance matrix of the $K$ non-zero components in $\mathbf{x}$.
\end{lemma}

\begin{lemma}\cite[page 249]{11:Marshall} \label{lem:constr bound}
	Let $\mathbf{A}$ and $\mathbf{B}$ are two $N \times N$ symmetric matrices, whose eigen-values $\alpha_1,\ldots, \alpha_N$ and $\beta_1,\ldots,\beta_N$ are ordered increasingly and decreasingly, respectively. Then $\Tr\{{\mathbf{AB}}\} \geq \sum_{i=1}^N \alpha_i \beta_i$.
\end{lemma}

\section{Proof of \theoref{theo:sing_ter}} \label{sec:app A}

To solve the optimization problem in \eqref{eq:opt 1}, let us first define
\begin{equation} \label{eq:MSE_S}
	\mathrm{MSE}^{(lb)}_\mathcal{S} \triangleq \Tr \left\{\left(\mathbf{R}^{-1} + g^2 \mathbf{E}_\mathcal{S}^\top \mathbf{H}^\top \mathbf{A}^\top \mathbf{R}_n^{-1} \mathbf{A} \mathbf{H} \mathbf{E}_\mathcal{S} \right)^{-1} \right\}.
\end{equation}
Using the matrix inversion lemma for $\mathbf{R}_n^{-1}$, we obtain
\begin{equation} \label{eq:mtx inv}
\begin{aligned}
	\mathbf{R}_n^{-1}  
		&= \sigma_w^{-2} \mathbf{I}_M - \sigma_w^{-2} \mathbf{A} \left(\frac{\sigma_w^{2}}{g^2 \sigma_v^{2}} \mathbf{I}_L + \mathbf{A}^\top \mathbf{A} \right)^{-1} \mathbf{A}^\top.& 
\end{aligned}
\end{equation}
Plugging \eqref{eq:mtx inv} back into \eqref{eq:MSE_S}, it follows that

{\small \begin{equation} \label{eq:opt 1_reform1}
\begin{aligned}
	&\mathrm{MSE}^{(lb)}_\mathcal{S} = \Tr \left\{\bigg(\mathbf{R}^{-1} + \frac{g^2}{\sigma_w^2} \mathbf{E}_\mathcal{S}^\top \mathbf{H}^\top \mathbf{A}^\top \mathbf{A} \mathbf{H} \mathbf{E}_\mathcal{S} \right.& \\
	&  \left. - \frac{g^2}{\sigma_w^2} \mathbf{E}_\mathcal{S}^\top \mathbf{H}^\top \mathbf{A}^\top  \mathbf{A}\left(\frac{\sigma_w^2}{g^2 \sigma_v^2} \mathbf{I}_L + \mathbf{A}^\top \mathbf{A}\right)^{-1} \mathbf{A}^\top \mathbf{AHE}_\mathcal{S} \bigg)^{-1} \right\}.& 
\end{aligned}
\end{equation}}

Next, defining $\mathbf{Q} \triangleq \mathbf{A}^\top \mathbf{A}$ and $\mathbf{D}_\mathcal{S} \triangleq \mathbf{HE}_\mathcal{S}$, the original optimization problem in \eqref{eq:opt 1} for finding optimized sensing matrix $\mathbf{A}$ can be equivalently translated into\footnote[1]{Note that since $p(\mathcal{S}) = 1/{N \choose K}$, it can be ignored in formulating the resulting optimization problems.}
{\small \begin{equation} \label{eq:opt 1_reform2}
\begin{aligned}
	&\underset{\mathbf{Q}}{\text{minimize}} \hspace{0.25cm} \sum_\mathcal{S} \Tr \left\{\bigg(\mathbf{R}^{-1} + \frac{g^2}{\sigma_w^2} \mathbf{D}_\mathcal{S}^\top \mathbf{Q} \mathbf{D}_\mathcal{S}  \right.& \\
	 &\left. \hspace{2cm}- \frac{g^2}{\sigma_w^2} \mathbf{D}_\mathcal{S}^\top  \mathbf{Q} \left(\frac{\sigma_w^2}{g^2 \sigma_v^2} \mathbf{I}_L + \mathbf{Q} \right)^{-1} \mathbf{Q} \mathbf{D}_\mathcal{S} \bigg)^{-1} \right\}& \\
	& \text{subject to} \hspace{0.25cm} \Tr \{(\mathbf{HR}_x \mathbf{H}^\top + \sigma_v^2 \mathbf{I}_L) \mathbf{Q}\} \leq P
	, \mathbf{Q} \succeq \mathbf{0} , \hspace{0.15cm} \mathrm{rank}(\mathbf{Q}) = M,&
\end{aligned}
\end{equation}}
where the rank constraint appears since $\mathbf{A} \in \mathbb{R}^{M \times L}$ with $M < L$. Introducing the semidefinite slack variable matrix $\mathbf{X}_\mathcal{S} \in \mathbb{R}^{K \times K}$, we can alternatively solve
\begin{equation} \label{eq:opt 1_reform3}
\begin{aligned}
	&\underset{\mathbf{Q},\mathbf{X}_\mathcal{S}}{\text{minimize}} \hspace{0.25cm} \sum_\mathcal{S} \Tr \{\mathbf{X}_\mathcal{S}\}	& \\
	&\text{subject to} \hspace{0.25cm} \bigg(\mathbf{R}^{-1} + \frac{g^2}{\sigma_w^2} \mathbf{D}_\mathcal{S}^\top \mathbf{Q} \mathbf{D}_\mathcal{S}  & \\
	 &\hspace{1.7cm} \!-\! \frac{g^2}{\sigma_w^2} \mathbf{D}_\mathcal{S}^\top  \mathbf{Q}\big(\frac{\sigma_w^2}{g^2 \sigma_v^2} \mathbf{I}_L\! + \!\mathbf{Q}\big)^{-\! 1} \mathbf{Q} \mathbf{D}_\mathcal{S} \bigg)^{- \!1} \! \! \! \! \! \preceq  \! \mathbf{X}_\mathcal{S}, \mathcal{S} \! \subset  \! \Omega& \\
	 &\hspace{1.7cm} \Tr \{ (\mathbf{HR}_x \mathbf{H}^\top + \sigma_v^2 \mathbf{I}_L) \mathbf{Q}\} \leq P & \\
	&\hspace{1.7cm} \mathbf{Q} \succeq \mathbf{0} , \hspace{0.15cm} \mathrm{rank}(\mathbf{Q}) = M.&
\end{aligned}
\end{equation}

Next, by applying the Schur's complement \cite{04:Boyd_book}, the first constraint in \eqref{eq:opt 1_reform3} can be rewritten as 
{\small \begin{equation} \label{eq:opt 1_reform4}
\begin{aligned}
&\left[\! \! \!
\begin{array}{c c}
   \mathbf{R}^{\!-1} \!+ \!\frac{g^2}{\sigma_w^2} \mathbf{D}_\mathcal{S}^\top \mathbf{Q} \mathbf{D}_\mathcal{S}  \!-\! \frac{g^2}{\sigma_w^2} \mathbf{D}_\mathcal{S}^\top  \mathbf{Q}(\frac{\sigma_w^2}{g^2 \sigma_v^2} \mathbf{I}_L \!+\! \mathbf{Q})^{-1} \mathbf{Q} \mathbf{D}_\mathcal{S} & \mathbf{I}_K \\ 
  \mathbf{I}_K  &   \mathbf{X}_\mathcal{S}   \\
\end{array} \! \! \!
\right] \succeq \mathbf{0} .& 
\end{aligned}
\end{equation}}

Introducing another slack semidefinite variable matrix $\mathbf{Y} \in \mathbb{R}^{L \times L}$, such that $\mathbf{Y} \succeq \frac{g^2}{\sigma_w^2}  \mathbf{Q}(\frac{\sigma_w^2}{g^2 \sigma_v^2} \mathbf{I}_N \!+\! \mathbf{Q})^{-1} \mathbf{Q}$, and using the Schur's complement for the resulting matrix inequality, we can further decompose the constraint in \eqref{eq:opt 1_reform4} into the two linear matrix inequalities in \eqref{eq:opt 1_final} which concludes the proof.

\section{Proof of \proref{prop:sepc_1}}
Using the notation $\mathbf{Q} = \mathbf{A}^\top \mathbf{A}$, we rewrite  \eqref{eq:opt 1_reform1} as
\begin{equation} \label{eq:closed_spec1_1}
\begin{aligned}
	&\mathrm{MSE}^{(lb)}_\mathcal{S} = \Tr \left\{\bigg(\frac{1}{\sigma_x^2} \mathbf{I}_K + \frac{g^2}{\sigma_w^2} \mathbf{E}_\mathcal{S}^\top \mathbf{Q} \mathbf{E}_\mathcal{S}  \right.& \\
	& \hspace{1.4cm} - \left. \frac{g^2}{\sigma_w^2} \mathbf{E}_\mathcal{S}^\top \mathbf{Q} \left(\frac{\sigma_w^2}{g^2 \sigma_v^2} \mathbf{I}_N + \mathbf{Q}\right)^{-1} \mathbf{Q} \mathbf{E}_\mathcal{S} \bigg)^{-1} \right\}.& 
\end{aligned}
\end{equation}

Applying \lemref{lem:deriv_cov}, the power constraint becomes
\begin{equation} \label{eq:power_cons_spec1}
\begin{aligned}
	&   \Tr \left\{ \left(\frac{\sigma_x^2}{K}   +  \sigma_v^2 \right) \mathbf{Q} \right\} \leq P,&
\end{aligned}
\end{equation}
and the objective function $\sum_\mathcal{S} \mathrm{MSE}^{(lb)}_\mathcal{S}$ is lower-bounded as
\begin{equation} \label{eq:closed_spec1_2}
\begin{aligned}
	&\sum_\mathcal{S} \mathrm{MSE}^{(lb)}_\mathcal{S} \geq \sum_\mathcal{S} K^2 \big /  \Tr \left\{\bigg(\frac{1}{\sigma_x^2} \mathbf{I}_K + \frac{g^2}{\sigma_w^2} \mathbf{E}_\mathcal{S}^\top \mathbf{Q} \mathbf{E}_\mathcal{S}  \right.& \\
	& \hspace{1.4cm} - \left. \frac{g^2}{\sigma_w^2} \mathbf{E}_\mathcal{S}^\top \mathbf{Q} \left(\frac{\sigma_w^2}{g^2 \sigma_v^2} \mathbf{I}_N + \mathbf{Q}\right)^{-1} \mathbf{Q} \mathbf{E}_\mathcal{S} \bigg)^{-1} \right\},&
\end{aligned}
\end{equation}
where we used the inequality $\Tr\{\mathbf{B}^{-1}\} \geq \frac{K^2}{\Tr\{\mathbf{B}\}}$ for a positive definite matrix $\mathbf{B}$ of dimension $K \times K$ \cite[Lemma 2]{03:Shengli}, in which the equality is satisfied when $\mathbf{B}$ becomes a scaled identity matrix. Hence, the objective function in the left hand side of \eqref{eq:closed_spec1_2} reaches its minimum when $\mathbf{Q} = \alpha \mathbf{I}_N$ (for some $\alpha > 0$) since $\mathbf{E}_\mathcal{S}^\top  \mathbf{E}_\mathcal{S} = \mathbf{I}_K$ (cf. \lemref{lem:prop}), and the matrix inside the argument of the trace becomes an scaled identity matrix. Note that this choice of $\mathbf{Q}$ does not affect the power constraint. Further, the coefficient $\alpha$ is derived such that the constraint  \eqref{eq:power_cons_spec1} is satisfied with equality that yields $\alpha = \frac{KP}{N(\sigma_x^2 + K \sigma_v^2)}$. Therefore, assuming $\mathbf{R} = \sigma_x^2 \mathbf{I}_K$ and  $\mathbf{H} = \mathbf{I}_N$, the solution to the SDR problem is 
\begin{equation} \label{eq:closed_spec1_3}
	\mathbf{Q}^\star = \frac{KP}{N(\sigma_x^2 + K \sigma_v^2)} \mathbf{I}_N. 
\end{equation}

Hence, the optimal sensing matrix $\mathbf{A}$ (with respect to \eqref{eq:opt_rec_A_appx}), after rescaling to meet the power constraint, becomes \eqref{eq:closed_spec2_4}.

\section{Proof of \proref{prop:sepc_2}}
Following the assumption in \proref{prop:sepc_2}, the SDR optimization problem simplifies into
\begin{equation} \label{eq:spec2_ref_1}
\begin{aligned}
	&\underset{\mathbf{Q}  }{\text{minimize}} \hspace{0.25cm} \sum_\mathcal{S} \Tr \left\{\bigg(\frac{1}{\sigma_x^2} \mathbf{I}_K + \frac{g^2}{\sigma_w^2} \mathbf{E}_\mathcal{S}^\top \mathbf{H}^\top \mathbf{Q} \mathbf{H} \mathbf{E}_\mathcal{S}  \bigg)^{-1} \right\}& \\
	& \text{subject to} \hspace{0.25cm} \frac{\sigma_x^2}{K}  \; \Tr \{\mathbf{H}^\top  \mathbf{Q} \mathbf{H}\} \leq P .& 
\end{aligned}	
\end{equation}

The objective function in \eqref{eq:spec2_ref_1} reaches its minimum when $\mathbf{H}^\top \mathbf{Q} \mathbf{H} = \alpha \mathbf{I}_N$ (see \cite[Lemma 2]{03:Shengli}). Taking SVD, we have $\mathbf{H} = \mathbf{U}_H \mathbf{\Gamma}_H  \mathbf{V}_H^\top$, where $ \mathbf{U}_H$ and  $\mathbf{V}_H$ are $N \times N$ unitary matrices and $\mathbf{\Gamma}_H = \mathrm{diag}(\gamma_{h_1},\gamma_{h_2}, \ldots, \gamma_{h_N} )$ is a diagonal matrix containing singular values $\gamma_{h_1}<\gamma_{h_2}< \ldots < \gamma_{h_N}$. Then, it follows that $\mathbf{Q}^\star$ should have the following structure
\begin{equation} \label{eq:spec2_ref_2}
	\mathbf{Q}^\star = \alpha (\mathbf{HH}^\top)^{-1} =  \alpha \mathbf{U}_H \mathbf{\Gamma}_H^{-2} \mathbf{U}_H^\top,
\end{equation}
where by plugging \label{eq:spec2_ref_2} into the power constraint, we obtain $\alpha = \frac{KP}{N \sigma_x^2}$. Therefore, the optimal sensing matrix $\mathbf{A}$ (with respect to \eqref{eq:opt_rec_A_appx}) can be chosen as in \eqref{eq:closed_spec2_4} after power rescaling.

\section{Proof of \proref{prop:sepc_3}}
Having the assumptions in \proref{prop:sepc_3}, the oracle estimator in \eqref{eq:oracle MMSE est} becomes
\begin{equation} \label{eq:oracle_est_spe_3}
\widehat{\mathbf{x}}^{(or)} \!=\! 
	g \left(\frac{g^2 \sigma_v^2}{\sigma_x^2} \mathbf{I}_K \!+\! g^2   \mathbf{A}_\mathcal{S}^\top \; (\mathbf{A}\mathbf{A}^\top)^{\! \dagger} \; \mathbf{A}_\mathcal{S} \right)^{\!-1} \!\!\mathbf{A}_\mathcal{S}^\top   (\mathbf{A}\mathbf{A}^\top)^{\! \dagger} \mathbf{y},
\end{equation}
where $(\cdot)^\dagger$ denotes matrix pseudo-inverse. It yields
\begin{equation}
	\mathrm{MSE}^{(lb)} \!= \! \! \frac{1}{{N \choose K}} \! \sum_\mathcal{S} \Tr \! \left \{ \! \left(\frac{1}{\sigma_x^2} \mathbf{I}_K \!+\! \frac{1}{\sigma_v^2} \mathbf{E}_\mathcal{S}^\top \mathbf{A}^{\!\top} (\mathbf{A}\mathbf{A}^{\! \top})^{\! \dagger} \mathbf{AE}_\mathcal{S} \! \right)^{\!-1} \! \right \}.
\end{equation}
Taking SVD, $\mathbf{A} = \mathbf{U}_a [\mathbf{\Gamma}_a \:\: \mathbf{0}_{N-M}] \mathbf{V}_a^\top$, it follows that 
\begin{equation} \label{eq:svd A}
	\mathbf{A}^{\!\top} (\mathbf{A}\mathbf{A}^{\! \top})^{\! \dagger} \mathbf{A} = \mathbf{V}_a 
	\left[ \begin{array}{c c} 
	   \mathbf{I}_M  & \mathbf{0}_{M \times (N-M)} \\ 
	   \mathbf{0}_{(N-M) \times M}  &   \mathbf{0}_{(N-M)\times (N-M)}  \\
	\end{array} \right]
	\mathbf{V}_a^\top.
\end{equation}

Applying \eqref{eq:svd A} into \eqref{eq:oracle_est_spe_3}, we have the following problem
\begin{equation} \label{eq:spec3_opt_prob}
\begin{aligned}
	&\underset{\mathbf{\Gamma}_a, \mathbf{V}_a }{\text{minimize}} \hspace{0.25cm} \sum_\mathcal{S} \Tr \left\{\bigg(\frac{1}{\sigma_x^2} \mathbf{I}_K \!+\! \frac{1}{\sigma_v^2} \mathbf{E}_\mathcal{S}^\top \mathbf{V}_a \!
	\left[ \! \! \begin{array}{c c} 
	   \mathbf{I}_M  & \mathbf{0}\\ 
	   \mathbf{0} &   \mathbf{0}  \\
	\end{array} \! \! \right] \! \mathbf{V}_a^{\! \top}
	 \mathbf{E}_\mathcal{S}  \bigg)^{\!-1} \right\}& \\
	& \text{subject to} \hspace{0.25cm} \frac{\sigma_x^2}{K}  \; \Tr \{  \mathbf{\Gamma}_a^2 \} \leq P .& 
\end{aligned}	
\end{equation}

We note that the objective function in \eqref{eq:spec3_opt_prob} can be minimized with respect to $\mathbf{U}_a$ independent of $\mathbf{\Gamma}_a$ in the constraint. Now, since $\mathbf{E}_\mathcal{S}^\top \mathbf{V}_a \mathbf{V}_a^\top \mathbf{E}_\mathcal{S} = \mathbf{I}_K$, the objective function in \eqref{eq:spec3_opt_prob} can be lower-bounded as 
\begin{equation} \label{eq:rewrite_lb}
\begin{aligned}
	 & \sum_\mathcal{S} \Tr \left\{\bigg( \mathbf{E}_\mathcal{S}^\top \mathbf{V}_a \!
	\left[ \! \! \begin{array}{c c} 
	   (\frac{1}{\sigma_x^2} \!+\! \frac{1}{\sigma_v^2}) \mathbf{I}_M  & \mathbf{0}\\ 
	   \mathbf{0} &   \frac{1}{\sigma_x^2} \mathbf{I}_{N-M}  \\
	\end{array} \! \! \right] \! \mathbf{V}_a^{\! \top}
	 \mathbf{E}_\mathcal{S}  \bigg)_{ii}^{\!-1} \right\}  &
\end{aligned}
\end{equation}
where by $(\cdot)_{ii}$ we denote the diagonal elements of the corresponding matrix. The lower-bound in \eqref{eq:rewrite_lb} is satisfied with equality if and only if  the matrix inside the trace-inverse operator becomes diagonal, which yields $\mathbf{V}_a = \mathbf{I}_N$. Also, from the constraint in \eqref{eq:spec3_opt_prob}, it follows that $\mathbf{\Gamma}_a$ can be an arbitrary diagonal matrix satisfying the transmission power constraint. For simplicity, we set $\mathbf{\Gamma}_a = \sqrt{\frac{KP}{M( \sigma_x^2 + K\sigma_v^2)}} \mathbf{I}_M$. Hence, the optimal sensing matrix has the structure in \eqref{eq:spec_case_3}.

\section{Proof of \proref{prop:sepc_4} }
We have
\begin{equation} \label{eq:taylor ser}
\begin{aligned}
&	\mathrm{MSE}^{(lb)} = \frac{1}{{N \choose K}} \sum_\mathcal{S} \Tr \left\{\left(\mathbf{R}^{-1} + \frac{g^2}{\sigma_w^2} \mathbf{D}_\mathcal{S}^\top \mathbf{Q} \mathbf{D}_\mathcal{S} \right)^{-1} \right\} & \\ 
	&\overset{(a)}{=} \! \frac{1}{{N \choose K}} \! \sum_\mathcal{S} \Tr  \! \left\{\mathbf{R} \!-\! \frac{g^2}{\sigma_w^2}  \mathbf{R} \mathbf{D}_\mathcal{S}^\top \mathbf{Q} \mathbf{D}_\mathcal{S}  \mathbf{R} \right\} 
	\!+\! \mathcal{O}(\|\frac{g^2}{\sigma_w^2}   \mathbf{D}_\mathcal{S}^\top \mathbf{Q} \mathbf{D}_\mathcal{S} \|_F^2),&
\end{aligned}
\end{equation}
where $(a)$ follows from Taylor series for the inverse term inside the trace operator in the first equation. Since $\frac{g^2}{\sigma_w^2} \rightarrow 0$, then by neglecting the higher moments, the optimization problem in \eqref{eq:opt 1} can be asymptotically approximated as
\begin{equation} \label{eq:opt asymp}
\begin{aligned}
	&\underset{\mathbf{Q}}{\text{maximize}} \hspace{0.25cm} \sum_\mathcal{S} \Tr \left\{\mathbf{R} \mathbf{D}_\mathcal{S}^\top \mathbf{Q} \mathbf{D}_\mathcal{S}  \mathbf{R} \right\} & \\
	& \text{subject to} \hspace{0.25cm} \Tr \{\mathbf{HR}_x \mathbf{H}^\top \mathbf{Q}\} 
	\leq P &  \\
		&\hspace{1.7cm} \mathbf{Q} \succeq \mathbf{0} , \hspace{0.15cm} \mathrm{rank}(\mathbf{Q}) = M.&
\end{aligned}
\end{equation}

Defining the full-rank symmetric positive definite matrix $\mathbf{T} \triangleq \sum_\mathcal{S}  \mathbf{D}_\mathcal{S} \mathbf{R}^2 \mathbf{D}_\mathcal{S}^\top$, and denoting $\mathbf{T}^{1/2}\mathbf{Q}\mathbf{T}^{1/2} \triangleq \mathbf{L}$, the optimization problem in \eqref{eq:opt asymp} can be rewritten as 
\begin{equation} \label{eq:opt asymp 2}
\begin{aligned}
	&\underset{\mathbf{L}}{\text{maximize}} \hspace{0.25cm} \Tr \left\{\mathbf{L} \right\} & \\
	& \text{subject to} \hspace{0.25cm} \Tr \{\mathbf{T}^{-1/2} \mathbf{HR}_x \mathbf{H}^\top \mathbf{T}^{-1/2} \mathbf{L}\} 
	\leq P &  \\
		&\hspace{1.7cm} \mathbf{L} \succeq \mathbf{0},  \hspace{0.15cm} \mathrm{rank}(\mathbf{L}) = M.&
\end{aligned}
\end{equation}

Let $\mathbf{Z} \triangleq \mathbf{T}^{-1/2} \mathbf{HR}_x \mathbf{H}^\top \mathbf{T}^{-1/2}$ have the EVD $\mathbf{Z} = \mathbf{U}_z \mathbf{\Gamma}_z \mathbf{U}_z^\top$. We also decompose $\mathbf{L}$ as $\mathbf{L} = \mathbf{U}_l \mathbf{\Gamma}_l \mathbf{U}_l^\top$, where $\mathbf{U}_z$  and $\mathbf{U}_l$ are unitary matrices, and $\mathbf{\Gamma}_z$ and $\mathbf{\Gamma}_l$ are diagonal matrices containing $\gamma_{z_i}$ and $\gamma_{l_i}$, respectively. In order to solve \eqref{eq:opt asymp 2}, we drop the rank constraint, and relax \eqref{eq:opt asymp 2} using \lemref{lem:constr bound} as
\begin{equation} \label{eq:opt asymp 3}
\begin{aligned}
	&\underset{\{ \gamma_{l_i} \}_{i=1}^L}{\text{maximize}} \hspace{0.25cm} \sum_{i=1}^L \gamma_{l_i} & \\
	& \text{subject to} \hspace{0.25cm} \sum_{i=1}^L \gamma_{z_i} \gamma_{l_i}
	\leq P , \hspace{0.25cm} \gamma_{l_i} \geq 0 \; \; , \; \; 1 \leq i \leq L,& 
\end{aligned}
\end{equation}
where $\gamma_{l_1} \geq \ldots \geq \gamma_{l_L}$ and $\gamma_{z_1} \leq \ldots \leq \gamma_{z_L}$.

Note that the optimization problem \eqref{eq:opt asymp 2}, without the rank constraint, and \eqref{eq:opt asymp 3} become equivalent when $\mathbf{Z L}$ is diagonal. This holds when $\mathbf{U}_l = \mathbf{U}_z$, where the columns of $\mathbf{U}_z$ are associated with the eigen-values of $\mathbf{Z}$ in an increasing order. Now, it only remains to solve \eqref{eq:opt asymp 3}. It is well-known that the objective function in \eqref{eq:opt asymp 3} is maximized by letting $\gamma_{l_1} = \frac{P}{\gamma{z_1}}$, and $\gamma_{l_2} = \ldots = \gamma_{l_L} = 0$. Thus, it follows that
\begin{equation} \label{eq:sol asymp}
	\mathbf{Q}^\star =  \mathbf{T}^{-1/2}\mathbf{U}_z \mathrm{diag} \left(\frac{P}{\gamma_{z_1}},0,\ldots,0 \right)\mathbf{U}_z^\top \mathbf{T}^{-1/2}.
\end{equation}

From \eqref{eq:sol asymp}, it is observed that $\mathbf{Q}^\star$ has only one non-zero eigen-value. Using EVD of $\mathbf{Q}^\star$, we have $\mathbf{Q}^\star = \mathbf{U}_q  \mathrm{diag} \left(\gamma_q,0,\ldots,0 \right)  \mathbf{U}_q^\top$, where $\gamma_q > 0$ denotes the non-zero eigen-value of $\mathbf{Q}^\star$. Now, let the SVD of $\mathbf{A}$ be $\mathbf{A} = \mathbf{U}_a [\mathbf{\Gamma}_a \; \; \mathbf{0}_{M \times (L-M)}]\mathbf{V}_a^\top$, where $\mathbf{U}_a \in \mathbb{R}^{M \times M}$ and $\mathbf{V}_a \in \mathbb{R}^{L \times L}$ are unitary matrices, and $\mathbf{\Gamma}_a \in \mathbb{R}^{M \times M}$ is a diagonal matrix. From $\mathbf{Q} = \mathbf{A}^\top \mathbf{A}$, it is concluded that the optimal sensing matrix can be expressed as in \eqref{eq:sol asymp2}.

\section{Proof of \theoref{theo:mult_ter_ortho}}
Using the matrix inversion lemma, we obtain 
\begin{equation} \label{eq:mtx inv_dist}
\begin{aligned}
	\widetilde{\mathbf{R}}_n^{-1} \!  
	&= \!\mathrm{blkdiag} \! \left(\! \sigma_{w_1}^{\! -2} \mathbf{I}_{M_1} \! \!-\! \sigma_{w_1}^{-2} \mathbf{A}_1 \left(\frac{\sigma_{w_1}^{2}}{g_1^2 \sigma_{v_1}^{2}} \mathbf{I}_{L_1} \! \!+\! \mathbf{A}_1^\top \mathbf{A}_1 \right)^{\! -1} \!\mathbf{A}_1^\top, \right.& \\
	&\left. \sigma_{w_2}^{-2} \mathbf{I}_{M_2} \!-\! \sigma_{w_2}^{-2} \mathbf{A}_2\left(\frac{\sigma_{w_2}^{2}}{g_2^2 \sigma_{v_2}^{2}} \mathbf{I}_{L_2} \!+\! \mathbf{A}_2^\top \mathbf{A}_2 \right)^{-1} \mathbf{A}_2^\top \! \right).& 
\end{aligned}
\end{equation}
Defining $\widetilde{\mathbf{E}}_\mathcal{S} \triangleq \widetilde{\mathbf{H}} \mathbf{J}  \mathbf{E}_\mathcal{S}$, by plugging \eqref{eq:mtx inv_dist} back into \eqref{eq:oracle_MSE_dist}, it follows that
\begin{equation} \label{eq:MSE_lb_dist_final_ortho}
 \mathrm{MSE}_o^{(lb)} = \sum_\mathcal{S} \frac{1}{{N \choose K}} \Tr \left\{ \bigg(\mathbf{R}^{-1} \!+\! \widetilde{\mathbf{E}}_\mathcal{S}^\top \mathbf{S}  \widetilde{\mathbf{E}}_\mathcal{S} \!-\! \widetilde{\mathbf{E}}_\mathcal{S}^\top  \mathbf{T}  \widetilde{\mathbf{E}}_\mathcal{S} \bigg)^{-1}  \right\},
\end{equation}
where 
\begin{equation} \label{eq:def acc 2}
\begin{aligned}
	\mathbf{S} &\triangleq \mathrm{blkdiag}\left(\frac{g_1^2}{\sigma_{w_1}^2} \mathbf{Q}_1, \frac{g_2^2}{\sigma_{w_2}^2} \mathbf{Q}_2 \right)& \\
	\mathbf{T} &\triangleq \mathrm{blkdiag}\left(\frac{g_1^2}{\sigma_{w_1}^2} \mathbf{Q}_1 \left(\frac{\sigma_{w_1}^2}{g_1^2 \sigma_{v_1}^2} \mathbf{I}_{L_1} + \mathbf{Q}_1 \right)^{-1} \mathbf{Q}_1 \; ,  \right. &  \\
	&\hspace{1.75cm} \left. \frac{g_2^2}{\sigma_{w_2}^2} \mathbf{Q}_2  \left(\frac{\sigma_{w_2}^2}{g_2^2 \sigma_{v_2}^2} \mathbf{I}_{L_2} + \mathbf{Q}_2 \right)^{-1} \mathbf{Q}_2 \right),& \\
	\mathbf{Q}_1 &\triangleq \mathbf{A}_1^\top \mathbf{A}_1 , \; \; \mathbf{Q}_2 \triangleq \mathbf{A}_2^\top \mathbf{A}_2 .& 
\end{aligned}
\end{equation}

Introducing the semidefinite slack variable matrix $\widetilde{\mathbf{X}}_\mathcal{S} \in \mathbb{R}^{K \times K}$, we equivalently solve
\begin{equation} \label{eq:opt 1_reform3_ortho}
\begin{aligned}
	&\underset{\mathbf{Q}_l, \widetilde{\mathbf{X}}_\mathcal{S}}{\text{minimize}} \hspace{0.25cm} \sum_\mathcal{S} \Tr \{\widetilde{\mathbf{X}}_\mathcal{S}\}	& \\
	&\text{subject to} \hspace{0.25cm} \bigg(\mathbf{R}^{-1} \!+\! \widetilde{\mathbf{E}}_\mathcal{S}^\top \mathbf{S}  \widetilde{\mathbf{E}}_\mathcal{S} \!-\! \widetilde{\mathbf{E}}_\mathcal{S}^\top  \mathbf{T} \widetilde{\mathbf{E}}_\mathcal{S} \bigg)^{-1} \preceq  \widetilde{\mathbf{X}}_\mathcal{S}& \\
	 &\hspace{1.7cm} \sum_{l=1}^2 \Tr \{(\mathbf{H}_l \mathbf{R}_x \mathbf{H}_l^\top) \mathbf{Q}_l + \sigma_{v_l}^2 \mathbf{Q}_l\} \leq P & \\
	&\hspace{1.7cm} \mathbf{Q}_l \succeq \mathbf{0} , \hspace{0.15cm} \mathrm{rank}(\mathbf{Q}_l) = M_l, \;\; \forall{l}, \; \forall \mathcal{S} &
\end{aligned}
\end{equation}

Now, applying the Schur's complement, the first constraint in \eqref{eq:opt 1_reform3_ortho} can be rewritten as the positive semi-definite constraint 
\begin{equation} \label{eq:opt 1_reform4_ortho}
\begin{aligned}
&\left[\! \! \!
\begin{array}{c c}
 \mathbf{R}^{-1} +  \widetilde{\mathbf{E}}_\mathcal{S}^\top  \mathbf{S} \widetilde{\mathbf{E}}_\mathcal{S}  \!-\! \widetilde{\mathbf{E}}_\mathcal{S}^\top  \mathbf{T} \widetilde{\mathbf{E}}_\mathcal{S} & \mathbf{I}_K \\ 
  \mathbf{I}_K  &   \mathbf{X}_\mathcal{S}   \\
\end{array}
\right] \succeq \mathbf{0} .& 
\end{aligned}
\end{equation}

Since $\mathbf{T}$ is a block diagonal matrix, by introducing another two slack semidefinite variable matrices  $\mathbf{Y}_{1} \in \mathbb{R}^{L_1 \times L_1} , \mathbf{Y}_{2} \in \mathbb{R}^{L_2 \times L_2}$, the constraint \eqref{eq:opt 1_reform4_ortho} can be decomposed into 
\begin{equation*} \label{eq:opt 1_reform5_ortho}
\begin{aligned}
 &\left[\! \! \!
\begin{array}{c c}
 \mathbf{R}^{-1} \!+\!  \widetilde{\mathbf{E}}_\mathcal{S}^\top \mathbf{S}  \widetilde{\mathbf{E}}_\mathcal{S}  \!-\! \widetilde{\mathbf{E}}_\mathcal{S}^\top  \mathrm{blkdiag}(\mathbf{Y}_{1} , \mathbf{Y}_{2})   \widetilde{\mathbf{E}}_\mathcal{S} & \mathbf{I}_K \\ 
  \mathbf{I}_K  &   \mathbf{X}_\mathcal{S}   \\
\end{array}
\right] \succeq \mathbf{0} ,
\end{aligned}
\end{equation*}

\begin{equation*} \label{eq:opt 1_reform6_ortho}
\left[
\begin{array}{c c}
	   \mathbf{Y}_l & \frac{g_l}{\sigma_{w_l}}\mathbf{Q}_l \\ 
	  \frac{g_l}{\sigma_{w_l}}  \mathbf{Q}_l  &  \frac{\sigma_{w_l}^2}{g_l^2 \sigma_{v_l}^2} \mathbf{I}_{L_l}  + \mathbf{Q}_l \\
	\end{array}
	\right] \succeq \mathbf{0} , \; \; \forall l , \forall \mathcal{S},
\end{equation*}
which concludes the proof. 

\section{Proof of \theoref{theo:mult_ter_coh}}
We omit the proof of the theorem since it is similar to the proofs of \theoref{theo:sing_ter} and \theoref{theo:mult_ter_ortho} by introducing slack variables and by applying the Schur's complement. 

\end{appendices}

\bibliographystyle{IEEEtran}
\bibliography{IEEEfull,bibliokthPasha}

\end{document}